\documentclass[11pt]{article}
\usepackage[T1]{fontenc}
\usepackage{amsfonts}
\usepackage{amsmath}
\usepackage{amssymb}
\usepackage{amsthm}
\usepackage{bbm}
\usepackage{bm}
\usepackage{mathrsfs}
\usepackage{verbatim}
\usepackage{setspace}
\usepackage{pdfsync}
\usepackage{enumitem}
\usepackage{appendix}
\usepackage{nicefrac}
\usepackage{xcolor}
\addcontentsline{toc}{chapter}{APPENDICES}

\theoremstyle{plain}
\newtheorem{theorem}{Theorem}[section]
\newtheorem{proposition}[theorem]{Proposition}
\newtheorem{lemma}[theorem]{Lemma}

\theoremstyle{definition}
\newtheorem{definition}[theorem]{Definition}
\newtheorem{remark}[theorem]{Remark}

\newtheorem{example}[theorem]{Example}

\newtheorem{assumption}[theorem]{Assumption}
\newtheorem{notation}[theorem]{Notation}

\theoremstyle{remark}

\renewenvironment{thebibliography}[1]{%
\begin{oldthebibliography}{#1}%
\setlength{\baselineskip}{.9em}
\linespread{1}
\small
\setlength{\parskip}{0.3ex}%
\setlength{\itemsep}{.5em}%
}%
{%
\end{oldthebibliography}%
}

\newcommand{\E}{\mathbb{E}}
\newcommand{\F}{\mathbb{F}}

\newcommand{\I}{\mathbb{I}}

\newcommand{\N}{\mathbb{N}}
\renewcommand{\P}{\mathbb{P}}
\newcommand{\Q}{\mathbb{Q}}
\newcommand{\R}{\mathbb{R}}

\newcommand{\cA}{\mathcal{A}}

\newcommand{\cC}{\mathcal{C}}
\newcommand{\cD}{\mathcal{D}}

\newcommand{\cF}{\mathcal{F}}
\newcommand{\cG}{\mathcal{G}}

\newcommand{\cK}{\mathcal{K}}
\newcommand{\cL}{\mathcal{L}}

\newcommand{\cX}{\mathcal{X}}
\newcommand{\cZ}{\mathcal{Z}}

\newcommand{\ov}{\overline}

\numberwithin{equation}{section}

\usepackage[pdfborder={0 0 0}]{hyperref}
\hypersetup{
  urlcolor = black,
  pdfauthor = {Chong Liu, Ariel Neufeld, Philipp Harms},
  pdfkeywords = {supermartingale deflator; absence of a num\'eraire; NUPBR;\\ fundamental theorem of asset pricing; arbitrage of the first kind},
  pdftitle = {Supermartingale Deflators in the Absence of a Num\'{e}raire},
  pdfsubject = {Supermartingale Deflators in the Absence of a Num\'{e}raire},
  pdfpagemode = UseNone
}

\begin{document}

\title{\vspace{-1em}
Supermartingale Deflators\\ in the Absence of a Num\'{e}raire
\date{}
\author{
	Philipp Harms
	\thanks{Abteilung f\"ur Mathematische Stochastik, Albert-Ludwig Universit\"at Freiburg, \texttt{philipp.harms@stochastik.uni-freiburg.de}.}
	\and 
  Chong Liu
  \thanks{
  Mathematical Institute, University of Oxford, \texttt{chong.liu@maths.ox.ac.uk}.
  }
  \and
  Ariel Neufeld
   \thanks{
   Division of Mathematical Sciences, NTU Singapore, \texttt{ariel.neufeld@ntu.edu.sg}.
   }
   }
}
\maketitle \vspace{-1.2em}

\begin{abstract}
In this paper we study arbitrage theory of financial markets in the absence of a num\'eraire both in discrete and continuous time. 
In our main results, we provide a generalization of the classical equivalence between no unbounded profits with bounded risk (NUPBR) and the existence of a supermartingale deflator.
To obtain the desired results, we introduce a new approach based on disintegration of the underlying probability space into spaces where the market crashes at deterministic times.
\end{abstract}

\vspace{.9em}

{\small
\noindent \emph{Keywords:} supermartingale deflator; absence of a num\'eraire; NUPBR;\\ fundamental theorem of asset pricing; arbitrage of the first kind.

}

\section{Introduction}\label{sec:intro}

\paragraph*{Overview.}
A nearly universal assumption in the arbitrage theory of financial markets is the existence of a num\'eraire, i.e., the existence of a strictly positive traded asset. 
For instance, this assumption underlies the celebrated fundamental theorems of Delbaen, Schachermayer, Kabanov, and Kardaras \cite{DelbaenSchachermayer.98, Kabanov97, Kardaras.12}.
In practice, however, it is not always reasonable to make this assumption. 
Indeed, it may well happen in the presence of credit or systemic risk that all assets under consideration default in finite time. 
For instance, when a country defaults on its debt and issues a new currency, this devalues not only the domestic bond market but also its num\'eraire.
Markets with arbitrarily low negative interest rates are approximations of this situation. 
Moreover, a similar situation occurs in financial models where the assets within a defaultable market segment are quoted in terms of an index or market average, as can be reasonable in portfolio optimization or hedging.
Motivated by these examples, the purpose of this work is to study arbitrage theory of financial markets in the absence of a num\'eraire.

Our main result is a generalization of the classical equivalence between no unbounded profits with bounded risk (NUPBR) and the existence of a supermartingale deflator \cite{Kardaras.13.2}. 
NUPBR is a pivotal notion in arbitrage theory and a minimal requirement for reasonable financial models \cite{DelbaenSchachermayer.94, DelbaenSchachermayer.98, Karatzas.07, Kardaras.12, Kardaras.13.2}. 
It also plays a fundamental role in defining path-wise stochastic integrals in model-free finance \cite{Perkowski.16}.
To put it into context, at least for markets with num\'eraire, Delbaen and Schachermayer's condition of no free lunch with vanishing risk (NFLVR) is equivalent to NUPBR together with no-arbitrage (NA) \cite{DelbaenSchachermayer.94}. 
However, there are many reasonable models such as the three-dimensional Bessel process which satisfy NUPBR but violate NA.
Moreover, NUPBR is all that is needed for ensuring that expected utility maximization is well defined, and the maximizer is precisely the desired supermartingale deflator \cite{Karatzas.07}.
We show that a similar result holds for markets without num\'eraire.
Namely, in discrete time, NUPBR remains equivalent to the existence of a supermartingale deflator. 
In continuous time, this equivalence holds under an independence assumption on the time of the market crash. 
Without this independence assumption, the equivalence holds subject to additional boundedness conditions on the market, which can be rephrased equivalently as boundedness conditions on the deflator.
It remains open to what extent these additional conditions are really necessary.

\paragraph*{Arbitrage theory without num\'eraire.}
The construction of the supermartingale deflator in \cite{Kardaras.13.2} via maximization of the expected $\log$--utility presupposes the existence of a num\'eraire to ensure that the maximization problem is well--defined. 
This is not merely a shortcoming of the proof but turns out to be a fundamental problem, which requires several adaptations of classical definitions and arguments, as outlined next. 

First, the notion of NUPBR is too weak.
Recall that NUPBR is defined as boundedness in probability of the payoffs at the terminal time $T$. 
When there is a num\'eraire, this implies boundedness in probability of the payoffs at all intermediate times $t<T$, see \cite{Kardaras.13.2}. 
However, in the absence of a num\'eraire, the payoffs at intermediate times may be unbounded in probability, as e.g.\@ in Example~\ref{ex:NUPBR-C-T-NOT}, and this rules out the existence of a strictly positive supermartingale deflator. 

Second, the notion of fork convexity (also known as switching property) is too weak. According to the classical definition, fork convexity allows an agent to switch from any given asset to any other strictly positive asset. However, markets without num\'eraire may not contain any strictly positive asset at all. In this case fork convexity is trivially satisfied. The correct modification is to allow the agent to switch to a new asset contingent on the new asset being positive at the given time and state of nature, as spelled out in Definition~\ref{def:fork-convexity}.

Third, the following argument, which is crucial for the construction of a deflator in \cite[Theorem~2.3]{Kardaras.13.2}, breaks down: if the terminal payoff $X_T$ of an asset $X$ is optimal within the set of all terminal payoffs, then the payoff $X_t$ is optimal within the set of all payoffs at time $t$, for any intermediate time $t<T$. For example, this clearly does not hold on markets where all the terminal payoffs vanish identically. Additionally, some arguments in \cite{Kardaras.13.2} concerning the regularization of generalized supermartingales break down because they also rely on the existence of a num\'eraire.  

\newpage
\paragraph*{The time where the market crashes.}

Methodologically, this work relies heavily on an analysis of the first time $\tau$ where all assets in the market vanish or, more succinctly, the time $\tau$ where the market crashes. 
Loosely speaking, one may partition the scenario space $\Omega$ into disjoint subsets $\Omega_t$ where $\tau$ is constant and equal to $t$. 
On each slice $\Omega_t$, there exists a process which is strictly positive up to time $t$ and therefore can serve as a num\'eraire. 
Thus, one obtains under the classical conditions of \cite{Kardaras.13.2} a supermartingale deflator $Z^t$ on each space $\Omega_t$ endowed with the conditional probability measure. 
These local deflators $Z^t$ on $\Omega_t$ can then be pasted into a global deflator $Z$ on $\Omega$. 

This sketch can be turned rather directly into a rigorous proof if $\tau$ has countable support; see Theorem~\ref{theorem: second main theorem}.
Otherwise, the conditional probabilities (provided they exist) may be singular with respect to $\P$, and consequently NUPBR on $\Omega$ does not entail NUPBR on $\Omega_t$. 
To overcome these issues, we discretize time into a finite dyadic grid of $2^n$ intervals and apply the above pasting method there. 
This produces a strictly positive supermartingale deflator on the grid.
Passing to the limit $n\to\infty$ while preserving the strict positivity is the most important and difficult part of the paper. 
This requires good lower bounds on the deflators or, equivalently, good upper bounds on the assets.

\paragraph*{Previous literature.}
Previously, arbitrage theory for markets without num\'eraires has been studied only in finite discrete time by Tehranchi \cite{Tehranchi14}. 
However, recently a related preprint of B\'alint \cite{Balint20} on continuous-time markets, based on research independent of ours, has appeared.
The philosophy in \cite{Balint20}, as well as in the present paper, is as follows: one first localizes the market, then constructs local deflators on each localized piece, and finally pastes them together to get a global deflator. The essential difference is that the localization in \cite{Balint20} is performed in time, i.e., the terminal date $T$ is approximated from below by a sequence of stopping times $(T^n)_{n \ge 1}$ such that on each time horizon $[0,T^n)$ the market contains a num\'eraire.
In contrast, the methodology of this paper is to localize the market in ``space'' in the sense that the sample space $\Omega$ is partitioned into different parts such that there exists a num\'eraire under the corresponding conditional measure. Then it is intuitively clear that the approach in \cite{Balint20} works very well in continuous time but fails if the underlying processes are non-adapted to the given filtration; on the other hand, our techniques can handle non-adaptedness (in particular for finite discrete time markets) but need more technical assumptions to work in continuous time. Hence, we believe that both approaches can provide alternative and complementary perspectives for future research. Besides, we also give an explicit formula for constructing a deflator (see Theorems \ref{thm:discrete}, \ref{theorem: second main theorem} and \ref{theorem: third main theorem}), while the results in \cite{Balint20} are elegant but rather abstract.

\paragraph*{Structure of the paper.}
This paper is organized in the following way. In Section~\ref{sec:main} we introduce the setup, notations, and main results. In Section~\ref{subsec:proof-discrete} we prove the first main result in finite discrete time. Note that Tehranchi \cite{Tehranchi14} also proved a similar result, but our approach is quite different and provides an alternative perspective. In Section~\ref{subsec:proof-continuous} we consider markets in continuous time and find an
 equivalence 
condition for the existence of a supermartingale deflator.

\section{Setup and main results} \label{sec:main}
We fix a finite time horizon $T \in (0,\infty)$ and a filtered probability space $(\Omega, \cF, \F, \P)$. Moreover we let $\mathbb{I}\subseteq [0,T]$ be either $\I:=\{0,1,\dots, T\}$ for the discrete-time setting or $\mathbb{I}:= [0,T]$ for the continuous-time setting. 
Throughout the paper, we will use the following notation.
\begin{notation}
$ \mbox{}$ \\
$\bullet$ If not specified differently, every property of a random variable or a stochastic process such as, e.g., strict positivity, c\`{a}dl\`{a}g paths,... is understood to hold $\P$--a.s..\\
$\bullet$ We mean by a stochastic process $(X_t)_{t\in \I}$ simply a collection of $\cF$-measurable random variables. \\
$\bullet$  For any measure $\Q$ on $(\Omega,\cF)$ we denote  by $L^0(\Q)$ the set of all (equivalence classes of) random variables, which we endow with the metric which induces convergence in $\Q$-probability. Moreover, we denote by $L^0_+(\Q)\subseteq L^0(\Q)$ the set of nonnegative random variables and by $L^0_{++}(\Q)\subseteq L^0_+(\Q)$ the set of strictly positive random variables $X$ in the sense that $\Q[X>0]=1$. \\
$\bullet$ We call a set $\cC \subseteq L^0(\Q)$ to be \textit{$\Q$-bounded} or \textit{bounded in $L^0(\Q)$} if it is bounded in probability with respect to $\Q$, namely
\begin{equation*}
\lim_{M \to \infty} \sup_{X \in \cC} \Q\big[|X|\geq M\big]=0.
\end{equation*}
$\bullet$  Following \cite{Zitkovic.09} we say that a set $\cC\subseteq L^0_+(\Q)$ is \textit{$\Q$-convex compact} or \textit{convexly compact} if it is convex, closed, and $\Q$-bounded.\\
$\bullet$  Following \cite{Kardaras.13.2}, we say that a stochastic process $(X_t)_{t\in [0,T]}$ defined on $[0,T]$ is \textit{$\Q$-c\`adl\`ag} if the  mapping $[0,T]\ni t \mapsto X_t \in L^0(\Q)$ is right-continuous and has left-limits.
\end{notation}
\begin{definition}\label{def:wealth-set}
In the discrete-time setting, we call a collection of nonnegative processes, denoted by $\mathcal{X}$, a \textit{wealth process set} or \textit{market} on $\{0,1,\dots,T\}$ if it satisfies the following two conditions:
\begin{enumerate}
	\item Each $X\in \cX$  satisfies $X_0 = 1$,
	\item  for each $X \in \mathcal{X}$ we have that $X$ vanishes on the stochastic interval $[\![\tau^X,T]\!]$, where $\tau^X := \inf\{t \in \{0,1,\dots, T\} \mid X_t = 0\}$ with the convention $\inf \emptyset := \infty$.
\end{enumerate}	
In the continuous-time setting, we call a  collection of nonnegative  processes $\mathcal{X}$  a \textit{wealth process set} or \textit{market} on $[0,T]$ if it satisfies the following two conditions:
\begin{enumerate}
	\item Each $X\in \cX$  has c\`{a}dl\`{a}g paths and satisfies  $X_0 = 1$,
	\item  for each $X \in \mathcal{X}$ we have that  $X$ vanishes on the stochastic interval $[\![\tau^X,T]\!]$, where $\tau^X := \inf\{t \in [0,T] \mid X_t = 0 \text{ or } X_{t-} = 0\}$ with the convention $\inf \emptyset := \infty$ for all $X \in \mathcal{X}$.
\end{enumerate}
Furthermore, a wealth process set $\cX$ is called $\F$-adapted if each $X\in \cX$ is an $\F$-adapted process.
\end{definition}
\begin{definition}\label{def:numeraire}
	We call an element $X^{num}\in \cX$ a \emph{num\'eraire} for the market $\cX$ if $X^{num}_t$ is strictly positive for all time $t$.
\end{definition}
Our goal of this paper is to analyze markets which do not necessarily contain a num\'eraire, both in the case where the market $\cX$ is $\F$-adapted or not.

In the spirit of \cite{Zitkovic.02,Kardaras.13.2}, 
we introduce a notion of generalized fork convexity for wealth process sets.
\begin{definition}\label{def:fork-convexity}
We say that a wealth process set $\mathcal{X}$ defined on $\mathbb{I}$ satisfies the \textit{generalized fork convexity} if the following two conditions hold:
\begin{enumerate}
\item $\mathcal{X}$  is convex, i.e., $\lambda X^1 + (1 - \lambda)X^2 \in \mathcal{X}$ for any $\lambda \in [0,1]$, $X^1, X^2$ in $\mathcal{X}$,
\item for any $X^1,X^2,X^3$ in $\mathcal{X}$, $s \in \I$, and $ A \in \cF_s$, the process defined by
\begin{align}\label{eq: fork convexity in no numeraire case}
X_t := X^1_t \mathbbm{1}_{\{t < s\}} +\bigg[&\mathbbm{1}_A\Big(\big(\tfrac{X^2_t}{X^2_s}X^1_s\big)\mathbbm{1}_{\{X^2_s>0\}} + X^1_t\mathbbm{1}_{\{X^2_s=0\}}\Big)  
\nonumber\\
&
+\mathbbm{1}_{A^c}\Big(\big(\tfrac{X^3_t}{X^3_s}X^1_s\big)\mathbbm{1}_{\{X^3_s>0\}} + X^1_t\mathbbm{1}_{\{X^3_s=0\}}\Big)\bigg]\mathbbm{1}_{\{t \geq s\}}, \quad \quad t \in \I,
\end{align}
belongs to $\mathcal{X}$. 
\end{enumerate}
\end{definition}
In words, the generalized fork convexity means that the agent on this market will switch to another portfolio at time $t$ only when the wealth process associated to the new portfolio has a positive value at this instant, otherwise she will keep her original position.
\begin{remark}\label{remark: compare with the classical fork convexity}
	We  point out that our notion of fork convexity is slightly more general than the usual one introduced by {\v{Z}}itkovi{\'c} \cite{Zitkovic.02}
	and 
	also used in Karadaras \cite{Kardaras.13.2}, even if the market $\cX$ possesses a num\'eraire. 
	More precisely, in the notion of {\v{Z}}itkovi{\'c} \cite{Zitkovic.02}, 
	the switched portfolios $X^2$ and $X^3$  in \eqref{eq: fork convexity in no numeraire case} have to be strictly positive.
	Since in our work, we analyze markets which may not contain a num\'eraire, we believe that our slight generalization of fork convexity is the natural extension in that setting.
	To justify our notion, we observe that in the presence of a num\'eraire, the property for a market to satisfy NUPBR, meaning that the final value set $\cC_T:=\{X_T\colon X \in \cX \}$ is $\P$-bounded, does not depend on the choice of the definition of the fork convexity (between the one of {\v{Z}}itkovi{\'c} \cite{Zitkovic.02} and ours). More precisely, we have in Lemma~\ref{le:remark: compare with the classical fork convexity} the following result, whose proof we provide in the appendix:
\end{remark}
\begin{lemma}
	\label{le:remark: compare with the classical fork convexity}
	Let $\cX$ be a market which is {$\mathbb{F}$-adapted} and contains a num\'eraire and assume that it is fork convex in the sense of {\v{Z}}itkovi{\'c} \cite{Zitkovic.02}. Then the market $\cX$ satisfies the NUPBR condition if and only if its fork convex hull taken with respect to our notion (see Definition~\ref{def:fork-convexity}) satisfies the NUPBR condition.
\end{lemma}
In the spirit of \cite{Kardaras.13.2}, we introduce the notion of a (generalized) supermartingale deflator. 

\begin{definition}
We call a nonnegative stochastic process $(Y_t)_{t\in \I}$ a \emph{generalized supermartingale} on $\I$ if 
for all $s,t \in \I$ with $s\leq t$
\begin{equation*}\label{eq:generalized supermartingale property}
	\E_{\P}\Big[\tfrac{Y_t}{Y_s}\Big|\cF_s\Big] \le 1.
\end{equation*} 
\end{definition}

\begin{definition}\label{def:deflator-discrete}
	We call a nonnegative stochastic process $(Z_t)_{t\in \I}$ a \emph{generalized supermartingale deflator} on $\I$ for $\mathcal{X}$ if $Z_0 \leq 1$ and $ZX$ is a generalized supermartingale for all $X \in \mathcal{X}$, 
	i.e.\ for all $s,t \in \I$ with $s\leq t$
	\begin{equation}\label{eq:generalized supermartingale deflator property}
	\E_{\P}\Big[\tfrac{X_tZ_t}{X_sZ_s}\Big|\cF_s\Big] \le 1.
	\end{equation}
	Moreover, when the market is $\F$-adapted, we call $(Z_t)_{t\in \I}$ a \emph{supermartingale deflator} if $(Z_t)_{t\in \I}$ is additionally $\F$-adapted.
\end{definition}
\begin{remark}\label{remark: well defined condition}
	In the above equation~\eqref{eq:generalized supermartingale deflator property} we apply the convention that $0/0 := 0$. Thanks to the property~(ii) of a wealth process set, we have $\{X_s = 0\} \subseteq \{X_t = 0\}$ for $s\leq t$, which ensures that the formulation~\eqref{eq:generalized supermartingale deflator property} is well--defined. 
	To rule out trivialities, we are interested in the existence of \textit{strictly positive} (generalized) supermartingales.
\end{remark}
In a market $\cX$ which is fork convex (in the sense of \cite{Zitkovic.02}) and possesses a num\'eraire,
Kardaras has proven in \cite[Theorem~2.3]{Kardaras.13.2} the equivalence between $\cX$ satisfying the NUPBR condition
and the existence of a strictly positive, $\P$-c\`adl\`ag generalized supermartingale deflator.
It is natural to ask the question if this equivalence also holds true for  a market $\cX$ satisfying the (generalized) fork convexity property, but which does not possess a num\'eraire. It turns out that this equivalence fails when a num\'eraire is absent, as shown in the following example.

\begin{example}\label{ex:NUPBR-C-T-NOT}
	The following market satisfies NUPBR but does not admit any strictly positive generalized supermartingale deflator. In a continuous-time setting, let $T=1$ and consider for each $n\in \N$ the deterministic process $X^n$ defined by
	 \begin{equation}\label{eq:ex:NUPBR-C-T-NOT}
	X^n_t:= \min\!\big\{1+(2n-2)t, 2n-2nt \big\}, \quad t \in [0,1].
\end{equation}
In other words, for each $n \in \N$ the process $X^n$ is linear between $1$ and $n$ on the time interval $[0,\frac{1}{2}]$ and linear between $n$ and $0$ on the time interval $[\frac{1}{2},1]$. Let $\mathcal{X}$ be the fork convex hull of all $X^n$. This market satisfies NUPBR because the $T$--value set $\mathcal{C}_T = \{0\}$ is $\P$-bounded. 
However for each $t \in (0,1)$, we have by \eqref{eq:ex:NUPBR-C-T-NOT} that $\sup_{n\in \N} X^n_t=\infty$, hence, as each $X^n$ is deterministic, 
the $t$--value set $\mathcal{C}_t$ is not $\P$-bounded. This in turn contradicts the existence of a strictly positive generalized supermartingale deflator, which would enforce the $\P$-boundedness of $\cC_t$ for all $t$.
\end{example}

\begin{remark}\label{rem:C-T-C-t-NOT}
As pointed out in \cite{Kardaras.13.2}, note that when considering a market possessing a num\'eraire which satisfies the fork convexity, the $\P$-boundedness of the final value set $\cC_T:=\{X_T\colon X \in \cX \}$ is equivalent to the $\P$-boundedness of all the intermediate value set $\cC_t:=\{X_t\colon X \in \cX \}$ for all $t$.  This equivalence may fail when there is no num\'eraire, as shown in the above Example~\ref{ex:NUPBR-C-T-NOT}.
\end{remark}
The above discussions suggest to ask whether the existence of a strictly positive generalized supermartingale deflator is equivalent to the $\P$-boundedness of $\cC_t:=\{X_t\colon X \in \cX \}$ for all $t$. The latter property is the content of the following definition.

\begin{definition}\label{def:NUPBR}
A market $\cX$ satisfies the NUPBR$_t$ condition at time $t$ if the intermediate value set $\cC_t:=\{X_t\colon X \in \cX \}$ is $\P$-bounded. In particular, the NUPBR$_T$ condition for the final time $T$ coincides with the classical NUPBR condition. 
\end{definition}

It turns out that in the discrete-time setting (see Theorem~\ref{thm:discrete}) as well as under some additional structure on the market (see Theorem~\ref{theorem: second main theorem} and Theorem~\ref{theorem: independent clock}) the equivalence indeed holds. 
Moreover, in the general setting for the continuous-time case,  we provide in our main Theorem~\ref{theorem: third main theorem} a stronger condition than the $\P$-boundedness of all $\cC_t:=\{X_t\colon X \in \cX \}$ and show that this condition is indeed equivalent to the existence of a strictly positive, (generalized) supermartingale deflator.
\begin{remark}\label{rem:static-deflator-not-enough}
	At first glance, one could guess that the $\P$-boundedness of all $\cC_t:=\{X_t\colon X \in \cX \}$  should always ensure the existence of a strictly positive generalized supermartingale deflator  for the following reason. The $\P$-boundedness of all $\cC_t$ ensures that each $\cC_t$ is convexly compact, which in turn by \cite[Theorem~1.1]{Kardaras.10.2} ensures the existence of a maximal element $\widehat{f}_t$ with respect to the preference relation $\preceq$ defined by  $f \preceq g$ if and only if $\E_{\P}[f/g] \le 1$ with the convention $0/0 = 0$. However, note that compared to the classical case where a market contains a num\'eraire, see \cite[Theorem~3.2]{Kardaras.13.2}, one cannot guarantee that the process  $(\widehat{f}_t)$ is strictly positive and hence
	the process $(\nicefrac{1}{\widehat{f}_t})$ may not form a strictly positive generalized supermartingale deflator.

	Instead, we will see later that for markets which do not possess a num\'eraire,  the existence of a strictly positive generalized supermartingale deflator depends crucially on the behaviour of the process $(\widehat{f}_t)_{t \in \mathbb{I}}$ hitting zero. More precisely, we define the debut of $(\widehat{f}_t)_{t \in \mathbb{I}}$ at the origin:
	\begin{equation*}
	\tau = \inf\big\{t \in \mathbb{I}: \widehat{f}_t = 0\big\},
	\end{equation*}
	with the convention $\inf \emptyset := \infty$. In view of the property that $\E_{\P}[\nicefrac{f}{\widehat{f}_t}] \le 1$ for all $f \in \mathcal{C}_t$, we indeed have that $\{\widehat{f}_t = 0\} \subseteq \{X_t = 0\}$ for all $X \in \mathcal{X}$, which in turn implies that after time $\tau$, the whole market becomes extinct, or in other words, the market $\mathcal{X}$ only survives on $[0,\tau)$. Assume for the moment that $\tau$ is measurable (we refer to Subsection~\ref{subsec:main-continuous} for the precise conditions) and denote by $\cL(\tau)$ the distribution of $\tau$ on $[0,T] \cup \{\infty\}$. One of the crucial observations in this paper is that the support $\cL(\tau)$ determines conditions for the existence of a strictly positive generalized supermartingale deflator; we refer to Theorem~\ref{theorem: second main theorem}, Theorem~\ref{theorem: third main theorem}, and Theorem~\ref{theorem: independent clock}.
\end{remark}

\newpage
\subsection{Main results in discrete-time}\label{subsec:main-discrete}
\begin{theorem}\label{thm:discrete}
	Let $\cX$ be a 
	market
	 satisfying the generalized fork convexity property. Then the following two statements are equivalent:
	\begin{enumerate}
	 \item NUPBR$_t$ holds for every $t$, i.e., the set $\mathcal{C}_t := \{X_t : X \in \mathcal{X} \}$ is bounded in probability for every $t$.
	\item There exists a strictly positive generalized supermartingale deflator. 
	\end{enumerate}
If we assume in addition that the market is $\F$-adapted, then 
 the following two statements are equivalent.
\begin{enumerate}
	\item NUPBR$_t$ holds for every $t$.
	\item There exists a strictly positive  supermartingale deflator. 
\end{enumerate}
\end{theorem}
The proof of Theorem~\ref{thm:discrete} is provided in Subsection~\ref{subsec:proof-discrete}.
\subsection{Main results in continuous-time}\label{subsec:main-continuous}
In this subsection, we provide our main results in the continuous-time setting. 
Let us   first introduce the notion of a generalized num\'eraire.
\begin{definition}\label{def: generalized numeraire}
	A process $\overline{X} \in \mathcal{X}$ is called a \textit{generalized num\'eraire} if for every $t \in [0,T]$ and $X \in \mathcal{X}$, one has $\P\big[\{X_t > 0\} \cap \{\overline{X}_t = 0\}\big] = 0$. 
\end{definition}

Note that generalized num\'eraires are not required to be strictly positive. However, if a market possess a num\'eraire, then the notions of generalized num\'eraire and (classical) num\'eraire coincide.
Financially speaking, a generalized num\'eraire is an asset which can only default if the whole market defaults. As a possible example, one may consider a government bond of a country which has AAA sovereign credit rating.
In the present continuous-time setting, we assume that a generalized num\'eraire exists. 
\begin{assumption}\label{ass:gener-numeraire}
The market $\cX$ contains a generalized num\'eraire $\ov X \in \cX$.
\end{assumption}
In addition, we impose the following standing condition on the filtration. 
\begin{assumption}\label{ass:usual conditions}
The filtered probability space $(\Omega,\cF, \F,\P)$ satisfies the usual conditions, meaning that $\cF$ is $\P$-complete, each $\cF_t$ is $\P$-$\cF$-complete, and $\F$ is right-continuous. 
\end{assumption}
This standard assumption guarantees the existence of c\`adl\`ag versions of supermartingales; see \cite[Theorem~VI.4, p.69]{DellacherieMeyer.82}. Moreover, in the presence of a generalized num\'eraire $(\overline X_t)$, which by definition satisfies $\{\overline{X}_t = 0\} \subseteq \{X_t = 0\}$ $\P$--a.s.\ for all $X \in \cX$, this assumption guarantees that the following debut $\tau$ is $\cF$-measurable, see e.g.\ \cite[Theorem~III.44, p.64]{DellacherieMeyer.78}, since $\cF$ by assumption is $\P$-complete:
\begin{equation}\label{eq: tau in continuous time}
\tau := \inf\{t \in [0,T]: \overline{X}_t = 0\},
\end{equation}
using the convention $\inf\emptyset := \infty$.
This allows us to consider the distribution $\cL(\tau)$   of $\tau$ on $[0,T] \cup \{\infty\}$ whose support turns out to determine the conditions for the existence of a strictly positive generalized supermartingale deflator; see also Remark~\ref{rem:static-deflator-not-enough}. 
We first start with the result stating that as long as $\cL(\tau)$ is discrete, we obtain the desired equivalence between the existence of a strictly positive generalized supermartingale deflator and NUPBR$_t$ for all $t$, like in the discrete-time setting.
\begin{theorem}\label{theorem: second main theorem}
	Let the underlying probability space $(\Omega,\cF, \F,\P)$ satisfy the usual conditions, let $\mathcal{X}$ be a market 
	satisfying the generalized fork convexity, assume that $\cX$
	possesses a generalized num\'eraire $\ov X$, and let $\tau$ denote its debut  at zero as in \eqref{eq: tau in continuous time}. If the support of $\cL(\tau)$ only consists of atoms, then the two following statements are equivalent.
	\begin{enumerate}
	\item NUPBR$_t$ holds for every $t$, i.e., the set $\mathcal{C}_t := \{X_t : X \in \mathcal{X} \}$ is bounded in probability for every $t$.
	\item There exists a strictly positive, $\P$-c\`adl\`ag generalized supermartingale deflator. 
\end{enumerate}
If we assume in addition that the market is $\F$-adapted, then the  following two statements are equivalent:
	\begin{enumerate}
	\item NUPBR$_t$ holds for every $t$.
	\item There exists a strictly positive, c\`adl\`ag supermartingale deflator.
\end{enumerate}
\end{theorem}
\begin{remark}\label{rem:2nd-thm-to-classical}
Note that if $\mathcal{X}$ contains a num\'eraire, then $\cL(\tau) = \delta_{\{\infty\}}$, i.e., the distribution of $\tau$ is the Dirac measure at $\infty$. In this case, the above Theorem~\ref{theorem: second main theorem} coincides with the classical result of Kardaras in \cite[Theorem~2.3]{Kardaras.13.2} but  with respect to the fork convexity defined as in Definition~\ref{def:fork-convexity}, see also Lemma~\ref{le:remark: compare with the classical fork convexity} and Remark~\ref{rem:C-T-C-t-NOT}.
\end{remark}
The proof of Theorem~\ref{theorem: second main theorem} is similar to the one for  Theorem~\ref{thm:discrete} in the discrete-time case. Roughly speaking, the idea is to construct for each $t$ in the support of $\cL(\tau)$  a  strictly positive ``local supermartingale deflator'' 
 under each $\P[\,\cdot\,|\,\tau = t]$
in order to paste them into a global one. As we will see later, the difficulty of providing a characterization for the existence of a strictly positive generalized supermartingale deflator arises when the support of $\cL(\tau)$ contains an uncountable subset $J \subseteq [0,T]$. At first glance, one would like to follow the same approach as for the case where the support of $\cL(\tau)$ only consists of atoms. More precisely, assume that  regular conditional probabilities $\P[\,\cdot\,|\,\tau = t]$ for  $t \in J$ exist and one could  construct  for \emph{each} $t \in J$ a  strictly positive ``local supermartingale deflator'' 
under each $\P[\,\cdot\,|\,\tau = t]$, in order to paste them into a global one.
 However, since $J$ is uncountable, \emph{not} all of these conditional probabilities are  absolutely continuous with respect to $\P$. Therefore, the condition that NUPBR$_s$ holds for each $s$  may fail with respect to some $\P[\,\cdot\,|\,\tau = t]$, even if we impose it to hold with respect to $\P$, and as a consequence one cannot construct strictly positive ``local supermartingale deflators'' for those conditional probabilities.

 To overcome this technical difficulty we introduce a stronger condition than that NUPBR$_t$ for each $t$. This stronger condition roughly speaking requires $\cC_t$ to be bounded uniformly with respect to all conditional $\P[\,\cdot\,|\,\tau \in (r,u]\cap J]$ for any $r,u$ on a countable dense set. This condition effectively allows one to transfer the discrete-time argument to the general continuous-time setting and allows us to formulate in the following Theorem~\ref{theorem: third main theorem} a characterization of the existence of a strictly positive generalized supermartingale deflator. 
 To this end, we introduce the following notation, where we recall our standing assumption in the continuous-time setting that the market $\mathcal{X}$  possesses a generalized num\'eraire $\ov X$, whose debut at zero is denoted by $\tau$ as in \eqref{eq: tau in continuous time}.
\begin{notation}\label{not:J-A-cond-prob}
Let the market $\mathcal{X}$  possess a generalized num\'eraire $\ov X$ with debut $\tau$ and corresponding distribution $\cL(\tau)$. Then, from now on, we will use the following notation:

\vspace*{0.15cm}
\noindent
$\bullet$  $\mathcal{A}$ denotes  the collection of all atoms in the support of $\cL(\tau)$;\\
$\bullet$ $\mathcal{J} := \text{supp}(\cL(\tau)) \setminus \mathcal{A}$;\\
$\bullet$ for each $r<u \in [0,T]$ denote by $\Q_{r,u}: \cF \to [0,1]$ the map
\begin{equation*}
\Q_{r,u}[\,\cdot\,] :=
\begin{cases}
\P\big[\cdot\,\big|\,\tau \in (r,u]\cap \mathcal{J}\big] & \mbox{ if } \P\big[\tau \in (r,u]\cap \mathcal{J}\big]>0,\\
0 														& \mbox{ else. }
\end{cases}
\end{equation*} 
 \end{notation}
\begin{theorem}\label{theorem: third main theorem}
		Using Notation~\ref{not:J-A-cond-prob}, let the underlying probability space $(\Omega,\cF, \F,\P)$ satisfy the usual conditions, let $\mathcal{X}$ be an $\F$-adapted market satisfying the generalized fork convexity,  and assume that $\cX$ possesses a generalized num\'eraire, allowing for $\mathcal J \neq \emptyset$.
		Then 
		the following statements $(i)$ and $(ii)$ are equivalent:
	\begin{enumerate}
		\item The following two properties hold:
		\begin{enumerate}
		 \item NUPBR$_t$ holds for every $t\in [0,T]$.
		 \item  There exists a countable dense subset $\mathcal{D} \subseteq [0,T]$ containing $0$ and $T$ such that for every $t$ in $\mathcal{D}$
		 \begin{equation}\label{eq: uniform boundedness in probability}
		 \limsup_{M \rightarrow \infty}\sup_{r,u \in \mathcal{D}, u >r\ge t}\sup_{X_t \in \mathcal{C}_t}\Q_{r,u}\big[X_t \ge M\big] = 0.
		 \end{equation}
		 \end{enumerate}

		\item 
		The following two properties hold:
		\begin{enumerate}
		\item There exists a strictly positive c\`adl\`ag 
		 supermartingale deflator $(Z_s)_{s \in [0,T]}$ for $\mathcal{X}$. 
		\item  There exist a countable dense subset $\mathcal{D} \subseteq [0,T]$ containing $\{0,T\}$ and a strictly positive process $(Z^\infty_t)_{t \in \cD}$ with $Z^\infty_0 \le 1$ such that
		\begin{itemize}
			\item for all $s<t$ in $\cD$ and for all $X  \in \mathcal{X}$, 
			\begin{equation*}
		\E_{\P}\Big[\tfrac{X_tZ^\infty_t}{X_sZ^\infty_s}\Big] \le 1,
			\end{equation*}
			\item for all $s<t$ in $\cD$ and all $\Q_{r,u}$ with $t\leq r<u\in \cD$,
			\begin{equation*}
			\E_{\Q_{r,u}}\Big[\tfrac{X_t Z^\infty_t}{X_sZ^\infty_s}\Big] \le 1,
			\end{equation*}
			\item for all $t \in \mathcal{D}$, 
			\begin{equation}\label{eq: uniform boundedness of deflator}
			\limsup_{M \rightarrow \infty}\sup_{r,u \in \mathcal{D}, u >r\ge t}\Q_{r,u}\Big[\tfrac{1}{Z^\infty_t} \ge M\Big] = 0.
			\end{equation}
		\end{itemize} 
	\end{enumerate}

	\end{enumerate}
\end{theorem}
\begin{remark}\label{rem:main-thm-J-uncountable}
The property $\mathcal{J}\neq \emptyset$ can only happen if the market does not possess a num\'eraire, since under presence of a num\'eraire  $\text{supp}(\cL(\tau))=\{\infty\}$.
\end{remark}
\begin{remark}\label{rem:Interpretaion-main-result}
Conditions~\eqref{eq: uniform boundedness in probability} and~\eqref{eq: uniform boundedness of deflator} admit the following financial interpretation. By definition, the conditional measure $\Q_{r,u}$ models the investment possibilities of an informed trader who knows that the market crashes between times $r$ and $u$. Thus, the uniform integrability Condition~\eqref{eq: uniform boundedness in probability} means that such informational advantages do not aggregate into unbounded profit with bounded risk. Similarly, Condition~\eqref{eq: uniform boundedness of deflator} means that the deflator is bounded in probability uniformly with respect to all such insider information.
\end{remark}%
\begin{remark}\label{rem:main-thm-meaning}
In the discrete case, we will see from the proof of Theorem~\ref{thm:discrete}   that the strictly positive generalized supermartingale deflator is not only defined with respect to $\P$ but also all the conditional measures $\P[\,\cdot\,|\,\tau = t]$ for $t \in \{0,1,\dots,T\} \cup \{\infty\}$, see Remark~\ref{remark: global deflator is also local deflator}. In the continuous-time case the properties on $(Z^\infty_t)_{t \in \cD}$ can be roughly seen as the analogue.
\end{remark}

As discussed above in Remark~\ref{rem:static-deflator-not-enough}, the key property we need of $\tau$ defined in \eqref {eq: tau in continuous time} is that there exists a generalized num\'eraire which is strictly positive on $[\![0,\tau)\!)$, whereas the market dies out on $[\![\tau,T]\!]$. It turns out that if we can find a random time $\widetilde \tau$ which possesses exactly this property and is independent of the market, then we obtain the equivalence between the existence of a strictly positive, c\`adl\`ag supermartingale deflator 
and the property NUPBR$_t$ for all $t$.

\begin{theorem}\label{theorem: independent clock}
	Let the underlying probability space $(\Omega,\cF, \F,\P)$ satisfy the usual conditions, let $\mathcal{X}$ be an  $\F$-adapted market 
	satisfying the generalized fork convexity, and suppose that there exists an $\cF$--measurable random time $\widetilde\tau : \Omega \rightarrow (0,T]\cup \{\infty\}$ which satisfies the following two properties:
	
	\vspace*{0.15cm}
	\noindent
	$(1)$
	For each $s<t$ the event $\{\widetilde\tau \in (s,t]\}$ is independent of $\cF_s$.
	\\
	$(2)$ 
	{For each $s<t$ such that
		$\P\big[\widetilde\tau \in (s,t]\big] > 0$,
		the market $\mathcal{X}$  under $\P\big[\,\cdot\,\big|\, \widetilde\tau \in (s,t]\big]$ contains a num\'eraire until time $s$ and all elements in $\mathcal{X}$ vanish after time $t$.}
	
	\vspace*{0.15cm}
	\noindent
	Then  the following two statements are equivalent:
	\begin{enumerate}
		\item NUPBR$_t$ holds for every $t$.
		\item There exists a strictly positive, c\`adl\`ag supermartingale deflator.
	\end{enumerate}
\end{theorem}

\begin{remark}\label{rem:cond1-weak}
	Condition $(1)$
	means that at each time $s$, all the information given on the market modelled by $\cF_s$ does not provide any information when the market will crash in the future.
	This may be realistic in a situation where $\widetilde \tau$ models the appearance of an extreme event like a natural disaster.
\end{remark}

\newpage
\section{Proof of the main results}\label{sec:proof}
We first provide auxiliary results which will be used frequently in the proof of our main results. The following lemma is well known, but we provide its proof for the sake of completeness.
\begin{lemma}\label{lemma: closure and closure of solid hull is bounded}
	Let $\P$ be any probability measure on a measurable space, let $\cC$ be any subset in $L^0_+(\P)$ which is bounded in $\P$--probability and let $\Q\ll \P$. Then the following holds true.
	\begin{enumerate}
		\item The $\P$--closure $\mbox{cl}_\P(\cC)$ of $\cC$ is bounded in $L^0_+(\P)$.
		\item The $\P$--closure of the $\P$-solid hull $\mbox{sol}_\P(\cC)$ of $\cC$ is bounded in $L^0_+(\P)$, where the $\P$-solid hull $\mbox{sol}_\P(\cC)$ is defined by
		$\mbox{sol}_\P(\cC) := \{g \in L^0_+(\P): g \le f \ \P\text{--a.s.\ } \mbox{for some } f\in \cC\}$.
		\item The set $\cC$ is bounded in $\Q$--probability.
	\end{enumerate}
\end{lemma}
\begin{proof}
	(i) Let $f \in \mbox{cl}_\P(\cC)$ and let $f^n$, $n \in \N$, be a sequence in $\cC$ such that $f^n$ converges to $f$ in $\P$--probability. By passing to a subsequence, we can also assume that the sequence convergences $\P$--a.s. For a given $\varepsilon>0$, since $\cC$ is bounded in $L^0_+(\P)$, there is an $M > 0$ such that  $\P[f^n > M] \le \varepsilon$ holds for all $n$,
	which in turn implies by Fatou's lemma that $\P[f> M] \le \liminf_{n \rightarrow \infty} \P[f^n> M] \le \varepsilon$. Consequently, we conclude the $\P$--boundedness of $\mbox{cl}_\P(\cC)$.\\
	(ii) If $\cC$ is bounded in $L^0_+(\P)$, then by definition its $\P$-solid hull $\mbox{sol}_\P(\cC)$ is also bounded in $L^0_+(\P)$, and so by (i) is then its $\P$-closure.
	\\
	(iii). Let $\cZ \in L^1(\P)$ denote the Radon-Nikodym derivative of $\Q$ with respect to $\P$. Then  we have  
	for any $N,M>0$ 
	and $X \in \cC$ 
	that
	\begin{equation*}
	\begin{split}
	\Q[X\geq M] 
	&=\E_\P[\cZ\,\mathbbm{1}_{\{X\geq M\}}]\\
	&=
	\E_\P[\cZ\,\mathbbm{1}_{\{X\geq M\}}\,\mathbbm{1}_{\{Z> N\}}] + \E_\P[\cZ\,\mathbbm{1}_{\{X\geq M\}}\,\mathbbm{1}_{\{Z\leq N\}}]
	\\
	&\leq
	\E_\P[\cZ\,\,\mathbbm{1}_{\{Z> N\}}] + N \P[X\geq M].
	\end{split}
	\end{equation*} 
	 Therefore, using that by assumption $\cC$ is bounded in $\P$--probability, we see that for all $N>0$
	 \begin{equation*}
	 \lim_{M\to \infty}\sup_{X\in \cC}\Q[X\geq M] \leq  \E_\P[\cZ\,\mathbbm{1}_{\{Z> N\}}].
	 \end{equation*}
	 Letting now $N$ tend to infinity  implies (iii), as $\cZ \in L^1(\P)$.
\end{proof}
A crucial tool for the proof of our main results is the notion of \textit{static deflators} introduced by Kardaras \cite{Kardaras.10.2}. 
\begin{lemma}\label{lemma: static numeraire lemma in the absence of numeraire}
	Let $\mathcal{C} \subseteq L^0_+$ be convexly compact. Let $\preceq$ be a binary relation on $\mathcal{C}$ such that $f \preceq g$ if and only if $\E[f/g] \le 1$ with the convention $0/0 = 0$. Then there exists a unique maximal element $\widehat{f} \in \mathcal{C}$ with respect to this relation $\preceq$. In particular, it holds that  $\{\widehat{f} = 0\} \subseteq \{f=0\}$ $\P$--a.s.\ for all $f \in \mathcal{C}$.
\end{lemma}
\begin{proof}
	See the proof of \cite[Theorem~1.1]{Kardaras.10.2}.
\end{proof}
Note that the above lemma does not ensure that $\widehat{f}$ is strictly positive. However this follows from the stronger assumption  $\mathcal{C} \cap L^0_{++} \neq \emptyset$, as stated next.
\begin{lemma}\label{lemma: static deflator}(Theorem~3.2 in \cite{Kardaras.13.2})
	Let $\mathcal{C} \subseteq L^0_+$ be convexly compact 
	such that $\mathcal{C} \cap L^0_{++} \neq \emptyset$. Then there exists a unique  $\widehat{f} \in \mathcal{C} \cap L^0_{++}$ such that $\E[f/\widehat{f}] \le 1$ holds for all $f \in \mathcal{C}$, which is then called  the \textit{static deflator}.
\end{lemma}
The next proposition shows that the main result in Kardaras
\cite[Theorem~2.3]{Kardaras.13.2} remains valid also when a fork-convex wealth process set $\mathcal{X}$ does not contain a strictly positive process, but the closure $\mbox{cl}_\P(\cC_T)$ 
of the final value set $\mathcal{C}_T$ contains a strictly positive random variable.
\begin{proposition}\label{prop:Kardaras-3.2-closure-ok}
Let $\cX$ be a fork convex wealth process set
such  that $\mbox{cl}_\P(\cC_T)$ contains a strictly positive random variable. Then  the main result in Kardaras
\cite[Theorem~2.3]{Kardaras.13.2} remains valid, i.e., NUPBR is equivalent to the existence of a strictly positive, $\P$-c\`adl\`ag generalized supermartingale deflator.
\end{proposition}
\begin{proof}
A careful inspection shows that the proof of \cite[Theorem~2.3]{Kardaras.13.2} remains valid also in this slightly more general case, provided  we can show that $\mbox{cl}_\P(\cC_T) \cap L^0_{++}(\P) \neq \emptyset$ ensures that $\mbox{cl}_\P(\cC_t) \cap L^0_{++}(\P) \neq \emptyset$ for all $t < T$. This is shown in the remaining part.

Suppose that $g_T \in \mbox{cl}_\P(\cC_T)$ is strictly positive and let  $\xi^n$, $n \in \N$, be a sequence in $\mathcal{X}$ such that $\xi^n_T$ converges to $g_T$ in probability. By passing to a subsequence we can even require the convergence to hold $\P$--almost surely.  Since $\P[g_T > 0] = 1$, for each $\varepsilon > 0$ there is an $\eta > 0$ such that $\P[g_T > \varepsilon] \ge 1 - \eta$ and $\eta$ converges to $0$ if $\varepsilon \rightarrow 0$. Moreover, by Egorov's theorem, there exists a set $\Gamma_T \in \mathcal{F}$ with $\P[\Gamma_T] \ge 1 - \eta$  such that $\xi^n_T$ converges to $g_T$ uniformly on $\Gamma_T$. Hence it holds that $\P[\Gamma_T \cap \{g_T > \varepsilon\}] \ge 1 - 2\eta$ and there exists an $N$ such that for all $n \ge N$, $\xi^n_T > \frac{\varepsilon}{2} > 0$ on $\Gamma_T \cap \{g_T > \varepsilon\}$. Furthermore, thanks to the property (ii) of the wealth process set, we have $\xi^n_t > 0$ on $\Gamma_T \cap \{g_T > \varepsilon\}$ for all $n \ge N$. Now let $\widehat f_t \in \mbox{cl}_\P(\cC_t)$ be the maximal element in the sense of Lemma~\ref{lemma: static numeraire lemma in the absence of numeraire}. Then, we have 
\begin{equation*}
\P[\widehat f_t > 0] \ge \P[\xi^n_t >0] \ge \P[\Gamma_T \cap \{g_T > \varepsilon\}] \ge 1 - 2\eta.
\end{equation*}
 Since $\varepsilon$ (and hence $\eta$) can be chosen arbitrarily small, we conclude that $\P[\widehat f_t > 0] = 1$.
\end{proof}
\begin{remark}\label{rem:prop-Kardaras-Q-ok}
Using Lemma~\ref{lemma: closure and closure of solid hull is bounded}, we see that Proposition~\ref{prop:Kardaras-3.2-closure-ok} remains valid when replacing $\P$ by  any $\Q \ll \P$.
\end{remark}
The following lemma will be crucial to construct supermartingale deflators (i.e., $\F$-adapted ones).
\begin{lemma}\label{lemma: nice version of optimizer}
	Let $\Q$ and $\P$ be two probabilities such that $\Q \ll \P$ on $\cF$. Let $\mathcal{C}$ be any $\P$-bounded subset in $L^0(\Omega,\cG,\P)$ for some sub $\sigma$--field $\cG \subseteq \cF$. Then any $\Q$--random variable $g \in \mbox{cl}_{\Q}(\mathcal{C})$ admits a $\cG$--measurable 
	$\Q$-version $g^\prime$ which is in the $\P$-closure of the $\P$-solid hull of $\cC$.
\end{lemma}
\begin{proof}
	Since $g$ is in the closure of $\mathcal{C}$ inside $L^0(\Q)$, there exists a sequence $(g_n)_{n \in \N} \subseteq \mathcal{C} \subseteq L^0(\Q)$ such that $g = \lim_{n \rightarrow \infty} g_n$ \, $\Q$--a.s. 
	In addition, as  $\Q \ll \P$ and $\mathcal{C}$ is bounded in $\P$--probability, it is also bounded in $L^0(\Q)$ and consequently $g$ is finite $\Q$--a.s.\\
	Without loss of generality we can assume that each $g_n$ is a $\P$--random variable in $\mathcal{C}$ (i.e., we pick a representative of $g_n$ such that its equivalence class modulo $\P$--null sets belongs to $\mathcal{C} \subseteq L^0(\P)$); in particular, every $g_n$ is $\cG$--measurable. Then, consider the set 
	\begin{equation*}
	A:= \big\{\omega \in \Omega: g_n(\omega), n \in \N, \text{ is a Cauchy sequence in } \R\big\},
	\end{equation*}
	which is $\cG$--measurable.
	Since $g_n$ converges to a $\R$--valued random variable $g$ $\Q$--a.s., we have $\Q[A] = 1$. Now we define
	 \begin{equation*}
	g^\prime := \limsup_{n \rightarrow \infty} g_n \mathbbm{1}_A = \lim_{n \rightarrow \infty}g_n\mathbbm{1}_A.
	 \end{equation*} Clearly, we have the $\cG$--measurability of $g^\prime$ and  $\Q[g = g^\prime] = 1$. Hence, the equivalence class of $g^\prime$ modulo $\Q$--null set in $L^0(\Q)$ is equal to $g$. Moreover, by construction of $g'$, we see that $g'$ is in the $\P$-closure of the $\P$-solid hull of $\cC$. This completes the proof.
\end{proof}
\subsection{Proof of the main results in discrete-time}\label{subsec:proof-discrete}
In this subsection we provide the proof of our main result Theorem~\ref{thm:discrete} in the discrete-time setting. 
We first note that the most important tool in its proof is the concept of static deflators,  see Lemma~\ref{lemma: static numeraire lemma in the absence of numeraire}.
To visualize its importance, suppose that the NUPBR$_t$ condition holds for each $t$. Then  
the $\P$--closure $\mbox{cl}_\P(\cC_t)$ of $\cC_t$ is convexly compact and we can pick for every $t$ a maximal element $\widehat{f}_t  \in \mbox{cl}_\P(\cC_t)$. For this finite sequence $\hat{f}_t$, $t \in \{0,1,\dots, T\}$, we define $\tau$ as its first hitting time of $0$, namely
\begin{equation}\label{def: debut time tau}
\tau := \inf\big\{t \in \{0,1,\dots, T\}: \hat{f}_t = 0 \big\},
\end{equation}
with the convention $\inf \emptyset = \infty$. Observe that $\tau$ is $\cF$-measurable
and for each $t$ one has that
\begin{equation*}
\begin{split}
\{\tau = \infty\} &= \big\{\widehat{f}_s >0, \forall s \in \{0,1,\dots, T\}\big\}, \\
\{\tau = t\} &= \big\{\widehat{f}_s >0, \forall s<t, s \in \{0,1,\dots, T\}\big\} \cap \big\{\widehat{f}_t = 0\big\}. 
\end{split}
\end{equation*}
 It follows immediately that for all $t \in \{0,1,\dots, T\} \cup \{\infty\}$,
\begin{equation}\label{eq: before tau}
\forall s<t, s\in \{0,1,\dots, T\}, \quad \P\big[\widehat{f}_s > 0\,\big|\, \tau = t\big] = 1.
\end{equation}
Moreover, in view of Lemma~\ref{lemma: static numeraire lemma in the absence of numeraire}, we have $\P[X_t = 0\mid \tau = t] = 1$ for all $X \in \mathcal{X}$. As a consequence of the ``non--rebounce'' property (ii) of a wealth process set, one then has for each $t \in \{0,1,\dots, T\}$ that
\begin{equation}\label{eq: after tau}
\forall r \ge t, r\in \{0,1,\dots, T\}, \forall X \in \mathcal{X}, \quad \P[X_r = 0 \,|\,\tau = t] = 1.
\end{equation}
We will heavily take use of these properties of $\tau$ in the following proof of Theorem~\ref{thm:discrete}.
\begin{proof}[Proof of Theorem~\ref{thm:discrete}]
	We start  with the well-known direction $(ii)\Rightarrow (i)$, whose short proof we provide for the sake of completeness.  Suppose that $(ii)$ holds, i.e., there exists a strictly positive generalized supermartingale deflator  $Z$ for $\cX$. Then as $Z_0\leq 1$ by definition and  since $X_0=1$ for all $X \in \cX$, we see that \eqref{eq:generalized supermartingale deflator property} ensures that 
	\begin{equation}\label{eq:deflator-P-bound1}
	\E_{\P}[X_tZ_t]\leq 1.
	\end{equation}
	Therefore, Markov's inequality implies the $\P$-boundedness of the   set $\{X_tZ_t : X_t \in \cC_t\}$, which means that for any $\varepsilon > 0$, there exists an $M>0$ such that for all $X_t\in \cC_t$, it holds that $\P[X_tZ_t \ge M] \le \varepsilon$. In addition, note that
	\begin{equation*}
	\begin{split}
	\P[X_t\geq M^2]
	&= \P\big[X_t\geq M^2; Z_t\geq \tfrac{1}{M}\big] + \P\big[X_t\geq M^2; Z_t< \tfrac{1}{M}\big]\\
	&\leq \P\big[X_tZ_t\geq M\big] + \P\big[Z_t< \tfrac{1}{M}\big].
	\end{split}
	\end{equation*}
	Applying Markov's inequality again together with \eqref{eq:deflator-P-bound1} and the strict positivity of $Z_t$ hence assures that  we can pick $M$ large enough such that both $\P[X_tZ_t \ge M]\le \varepsilon$ and $\P[Z_t \le \frac{1}{M}] \le \varepsilon$ are satisfied. This in turn shows the $\P$-boundedness of the set $\mathcal{C}_t := \{X_t : X \in \mathcal{X} \}$ as desired.

	Now, let us prove the implication $(i) \Rightarrow (ii)$ and hence assume that  NUPBR$_t$ holds for each $t$.
	Then every $\mbox{cl}_\P(\cC_t)$ is convexly compact, which  by Lemma~\ref{lemma: static numeraire lemma in the absence of numeraire} allows us to choose a sequence of maximal elements $\widehat{f}_t \in \mbox{cl}_\P(\cC_t)$, $t \in \{0,1,\dots,T\}$, and can define the random
	time $\tau$ as in \eqref{def: debut time tau}. Without loss of generality, we may assume that all $\{\tau = t\}$, $t \in \{0,1,\dots,T\}$, have positive $\P$-measure (otherwise, we consider the subset $\{t \in \{0,1,\dots,T\} \cup \{\infty\}: \P[\tau = t] > 0\}$), and we introduce the notion $\Q_t[\,\cdot\,]$ to denote the conditional probability $\P[\,\cdot\,|\,\tau = t]$ for $t \in \{0,1,\dots,T\} \cup \{\infty\}$. We divide the proof of the implication $(i) \Rightarrow (ii)$
	into several steps.
	
	\textit{Step~1}: (local supermartingale deflators). 
	In this step we will show that for every $t \in \{0,1,\dots,T\}\cup \{\infty\}$ there is a  strictly positive generalized supermartingale deflator $(Y^{t}_s)_{s\in \{0,1,\dots,T\}}$ 
	with respect to the conditional probability $\Q_t$.
	First, let $t = \infty$. Since $\Q_\infty \ll \P$, the NUPBR$_t$ condition also holds with respect to $\Q_\infty$. In particular, by Lemma~\ref{lemma: closure and closure of solid hull is bounded} the set $\mbox{cl}_{\Q_\infty}(\cC_s)$ is convexly compact as a subset in $L^0_+(\Q_\infty)$ for each $s$.
	For every $s \in \{0,1,\dots,T\}$ we therefore obtain by Lemma~\ref{lemma: static numeraire lemma in the absence of numeraire} that there exists a unique maximal element $\widehat{f}_s^{\Q_{\infty}} \in \mbox{cl}_{\Q_\infty}(\cC_s)$. Moreover, in view of \eqref{eq: before tau}, there is a maximal element $\widehat{f}^{\P}_s \in \mbox{cl}_{\P}(\mathcal{C}_s)$ such that 
		$\P\big[\widehat{f}^{\P}_s > 0\,\big|\, \tau = \infty\big] = 1$. Since $\mbox{cl}_{\P}(\mathcal{C}_s) \subseteq \mbox{cl}_{\Q_\infty}(\mathcal{C}_s)$, we obtain by uniqueness of the maximal element that $\widehat{f}^{\P}_s=\widehat{f}_s^{\Q_{\infty}} \ \Q_{\infty}$-a.s., which we denote by $\widehat{f}_s$, and which satisfies that $\Q_\infty[\widehat{f}_s > 0] = 1$.
	 In other words, under the conditional measure $\Q_\infty$, the fork convex market $\mathcal{X}$ contains a num\'eraire and therefore, by the classical result proved by Kardaras in \cite[Theorem~2.3]{Kardaras.13.2} together with Proposition~\ref{prop:Kardaras-3.2-closure-ok} and Remark~\ref{rem:prop-Kardaras-Q-ok}, there exists a 
	generalized supermartingale deflator $Y^{\infty}$ with the following properties:  $Y^{\infty}$ is strictly positive with respect to $\Q_\infty$, $Y^{\infty}_0 \le 1$, and for all $X \in \mathcal{X}$, $s < r \in  \{0,1,\dots,T\}$, 
	\begin{equation}\label{eq: supermartingale property under tau = infty}
	\E_{\Q_\infty}\Big[\tfrac{X_rY^{\infty}_r}{X_sY^{\infty}_s}\Big| \cF_s\Big] \le 1.
	\end{equation}
	Next we consider $t = T$. Again, from \eqref{eq: before tau} we can conclude that strictly before time $T$ the market $\mathcal{X}$ satisfies the NUPBR$_t$ condition for all $t \in \{0,1,\dots,T-1\}$ and contains a num\'eraire with respect to the conditional measure $\Q_T$. Therefore, there exists 
	a strictly positive generalized supermartingale deflator $(Y^{T}_s)_{s\in \{0,1,\dots,T-1\}}$ defined on $\{0,1,\dots,T-1\}$. Furthermore, in view of \eqref{eq: after tau}, we have $\Q_T[X_T = 0] = 1$ for all $X \in \mathcal{X}$, which implies that $\E_{\Q_T}\big[\frac{X_T}{X_sY^{T}_s}\,\big|\,\cF_s\big] = 0 \le 1$ for all $s < T$. Hence, we can 
	extend $(Y^{T}_s)_{s\in \{0,1,\dots,T-1\}}$ from $\{0,1,\dots,T-1\}$ to $\{0,1,\dots,T\}$ by setting $Y^{T}_T := 1$, and the resulting process $(Y^{T}_s)_{s\in \{0,1,\dots,T\}}$ is 
	a strictly positive generalized supermartingale deflator for $\mathcal{X}$ with respect to $\Q_T$.
	We continue backwards for all $t \in  \{0,1,\dots,T\}$ with this  procedure. Then for every $t \in \{0,1,\dots,T\}$ we get 
	a strictly positive generalized supermartingale deflator $(Y^{t}_s)_{s\in \{0,1,\dots,T\}}$ with respect to $\Q_t$ such that for all $s \ge t$, $Y^{t}_s = 1$ whereas for all $s<t$, $Y^{t}_s$ is constructed as in the classical case by Kardaras \cite[Theorem~2.3]{Kardaras.13.2}, but under the measure $\Q_t$. In particular, for $s<t$, we know from \cite[Theorem~2.3]{Kardaras.13.2} that $Y^{t}_s$ can be obtained by the relation that
	\begin{equation}\label{eq: construction of local supermartingale deflator} 
	\tfrac{1}{Y^{t}_s} \in \mbox{cl}_{\Q_t}(\cC_s) \text{ is the \textit{static deflator} (cf.\ Lemma~\ref{lemma: static deflator}) in }\mbox{cl}_{\Q_t}(\cC_s) \text{ w.r.t.\ } \Q_t.
	\end{equation}  
	Of course, the supermartingale deflator property indeed gives us for any $s < r$  and $X \in \mathcal{X}$ that
	\begin{equation}\label{eq: supermartingale deflator property for every conditional measure}
	\E_{\Q_t}\Big[\tfrac{X_rY^{t}_r}{X_sY^{t}_s}\,\Big|\, \cF_s\Big] \le 1.
	\end{equation}
	
	\textit{Step~2}: (Pasting local deflators to a global one).
	In this step, we glue all $(Y^{t}_s)_{s \in \{0,1,\dots,T\}}$, $t \in \{0,1,\dots,T\} \cup \{\infty\}$, together to obtain 
	a strictly positive generalized supermartingale deflator for $\mathcal{X}$ under $\P$. More precisely, we define the process $(Z_s)_{s\in\{0,1,\dots,T\}}$ by setting for each $s$ 
	\begin{equation}\label{def: supermartingale deflator in finite case}
	Z_s := \sum_{t \in \{0,1,\dots,T\} \cup \{\infty\}} Y^{t}_s \, \mathbbm{1}_{\{\tau = t\}}, 
	\quad \quad s \in \{0,1,\dots,T\}.
	\end{equation}
	Then, for all $s<r \in \{0,1,\dots,T\}$, $X \in \mathcal{X}$, and $A_s \in \cF_s$, by \eqref{eq: supermartingale property under tau = infty}, \eqref{eq: supermartingale deflator property for every conditional measure}, \eqref{def: supermartingale deflator in finite case}, and the definition of $\Q_t[\,\cdot\,] := \P[\,\cdot\,|\,\tau = t]$, we  deduce that
	\begin{align*}
	\E_{\P}\Big[\tfrac{X_rZ_r}{X_sZ_s}\,\mathbbm{1}_{A_s}\Big] &= \sum_{t \in \{0,1,\dots,T\} \cup \{\infty\}} \E_{\Q_t}\Big[\tfrac{X_rZ_r}{X_sZ_s}\,\mathbbm{1}_{A_s}\Big]\P[\tau = t] \\
	&=\sum_{t \in \{0,1,\dots,T\} \cup \{\infty\}} \E_{\Q_t}\Big[\tfrac{X_rY^{t}_r}{X_sY^{t}_s}\,\mathbbm{1}_{A_s}\Big]\P[\tau = t] \\
	&\le \sum_{t \in \{0,1,\dots,T\} \cup \{\infty\}}\Q_t[A_s]\,\P[\tau = t] \\
	&= \P[A_s],
	\end{align*}
	which implies that $\E_{\P}\big[\frac{X_rZ^k_r}{X_sZ^k_s}\mid \cF_s\big] \le 1$.
	Furthermore, since by construction each $Y^{t}$ is strictly positive under $\Q_t$ for all $t \in \{0,1,\dots,T\} \cup \{\infty\}$, using a similar argument as above we can see that for all $s \in \{0,1,\dots,T\}$,
	\begin{align*}
	\P[Z_s > 0] &= \sum_{t \in \{0,1,\dots,T\} \cup \{\infty\}}\Q_t[Z_s > 0]\,\P[\tau = t] \\
	&= \sum_{t \in \{0,1,\dots,T\} \cup \{\infty\}} \Q_t[Y^{t}_s > 0]\,\P[\tau = t] \\
	&= \sum_{t \in \{0,1,\dots,T\} \cup \{\infty\}}\,\P[\tau = t] = 1.
	\end{align*}
	Hence, we conclude that the process $(Z_s)_{s\in\{0,1,\dots,T\}}$ defined in \eqref{def: supermartingale deflator in finite case} is indeed 
	a strictly positive generalized supermartingale deflator for $\mathcal{X}$ under $\P$. This finishes the proof for the non $\F$-adapted case.

	\textit{Step~3:} (Supermartingale deflator (i.e., $\F$-adapted one) for $\F$-adapted market).
	It remains to show that one can construct a  supermartingale deflator (i.e., an $\F$-adapted one) if the market is $\F$-adapted. 
	Using the notations introduced before, 
	we know from \eqref{eq: construction of local supermartingale deflator} that for any $t\in \{0,1,\dots,T\}$ there exists a strictly positive generalized supermartingale deflator $(Y^{t}_s)_{s \in \{0,1,\dots,T\}}$ with respect to the conditional measure $\Q_t$ which satisfies for all  $s<t$ that $\frac{1}{Y^{t}_s} \in \mbox{cl}_{\Q_t}(\mathcal{C}_s)$, whereas $Y^{t}_s = 1$ for all $s \ge t$. Now, for any $s<t$, since the market is $\F$-adapted, we can apply Lemma~\ref{lemma: nice version of optimizer} with $\Q := \Q_t$ and $g := \frac{1}{Y^{t}_s}$ to find an $\cF_s$--measurable $\Q_t$-version of $\frac{1}{Y^{t}_s}$, which we still denote $\frac{1}{Y^{t}_s}$ for the ease of notation. From this construction, 
	we obtain an $\F$--adapted 
	strictly positive generalized supermartingale deflator with respect to $\Q_t$
	which we again denote by $(Y^{t}_s)_{s \in \{0,1,\dots,T\}}$ for the ease of notation. Consequently, the generalized supermartingale property \eqref{eq: supermartingale deflator property for every conditional measure} can be rewritten as the standard supermartingale property, namely for all $X \in \mathcal{X}$ and $r\geq s$
	\begin{equation*}
	\E_{\Q_t}\big[X_rY^{t}_r\big| \cF_s\big] 
	\le
	 X_sY^{t}_s,
		\end{equation*}
	or, equivalently, for all $A_s \in \cF_s$,
	\begin{equation*}
	\E_{\Q_t}\big[X_rY^{t}_r\,\mathbbm{1}_{A_s}\big] 
	\le
	 \E_{\Q_t}\big[X_sY^{t}_s\,\mathbbm{1}_{A_s}\big].
	\end{equation*}
	Now, let $(Z_s)_{s\in\{0,1,\dots,T\}}$ be the process defined in \eqref{def: supermartingale deflator in finite case} and let $(\widetilde Z_s)_{s\in\{0,1,\dots,T\}}$ be the process defined by setting  $\widetilde{Z}_s := \E_{\P}[Z_s\,|\, \cF_s]$ for each $s \in \{0,1,\dots,T\}$. Then we get that  
	\begin{align*}
	\E_{\P}\big[X_r\widetilde{Z}_r\,\mathbbm{1}_{A_s}\big] 
	=
	 \E_{\P}\big[X_rZ_r\,\mathbbm{1}_{A_s}\big] 
	 &=
	  \sum_{t \in \{0,1,\dots,T\} \cup \{\infty\}} \E_{\Q_t}\big[X_rY^{t}_r\,\mathbbm{1}_{A_s}\big]\, \P[\tau = t] \\
	&\le \sum_{t \in \{0,1,\dots,T\} \cup \{\infty\}} \E_{\Q_t}\big[X_sY^{t}_s\,\mathbbm{1}_{A_s}\big] \,\P[\tau = t] \\
	&= \E_{\P}\big[X_sZ_s\,\mathbbm{1}_{A_s}\big] = \E_{\P}\big[X_s\widetilde{Z}_s\,\mathbbm{1}_{A_s}\big].
	\end{align*}
	Since by definition every $\widetilde{Z}_s$ is $\cF_s$--measurable, the above inequality indeed shows that $(\widetilde Z_s)_{s\in\{0,1,\dots,T\}}$ is 
	a strictly positive supermartingale deflator for $\mathcal{X}$ under $\P$. This completes the proof of Theorem~\ref{thm:discrete}.
\end{proof}
\begin{remark}\label{remark: global deflator is also local deflator}
	The construction of the 
	generalized supermartingale deflator $(Z_s)_{s\in\{0,1,\dots,T\}}$ in the proof of Theorem~\ref{thm:discrete}, see \eqref{def: supermartingale deflator in finite case}, shows that $(Z_s)_{s\in\{0,1,\dots,T\}}$ is not only
	 a strictly positive generalized 
	supermartingale deflator under $\P$, but also
	simultaneously under all $\Q_t[\,\cdot\,] := \P[\,\cdot\,|\,\tau = t]$ for $t \in \{0,1,\dots,T\} \cup \{\infty\}$. This observation turns out to be important in the proof of Theorem~\ref{theorem: third main theorem}.
\end{remark}
\subsection{Proof of the main results in continuous time}\label{subsec:proof-continuous}
\subsubsection{Proof of Theorem~\ref{theorem: second main theorem}}\label{subsubsec:second-main-thm}
In this subsubsection we provide the proof of Theorem~\ref{theorem: second main theorem}, which is essentially the same as the one of Theorem~\ref{thm:discrete}.
\begin{proof}[Proof of Theorem~\ref{theorem: second main theorem}]
	First, note that the implication $(ii) \Rightarrow (i)$ follows by the same argument as, e.g., in the proof of Theorem~\ref{thm:discrete}. Hence we now show that the implication $(i) \Rightarrow (ii)$ holds.
	
	To that end, for each $t \in \cA$ we denote $\Q_t[\,\cdot\,]:= \P[\cdot \,|\, \tau = t]$. Since $(i)$ ensures that the NUPBR$_s$ condition holds for each $s\in [0,T]$ under the measure $\P$, and as all $\Q_t$, $t \in \cA$, are absolutely continuous with respect to $\P$, the market $\mathcal{X}$ satisfies the NUPBR$_s$ condition for each $s\in [0,T]$ also under every $\Q_t$, $t \in \cA$.
	Moreover, since $\{\tau = t\} = \{\overline{X}_s > 0, \forall s<t\} \cap \{\overline{X}_t = 0\}$, we have similarly to the discrete-time case \eqref{eq: before tau} and \eqref{eq: after tau}, that for every $t \in \cA$,
	\begin{equation}\label{eq: before tau in continuous time}
	\forall s<t, s \in [0,T], \quad \Q_t [\overline{X}_s >0] = 1,
	\end{equation}
	and 
	\begin{equation}\label{eq: after tau in continuous time}
	\forall s\ge t, s \in [0,T], \forall X \in \mathcal{X}, \quad \Q_t[X_s =0] = 1.
	\end{equation}
	 Next, for each $t \in \cA \cap [0,T]$, in view of \eqref{eq: before tau in continuous time} and \eqref{eq: after tau in continuous time} we have
	\begin{enumerate}
		\item The market $\mathcal{X}$ contains a num\'eraire strictly before time $t$ with respect to $\Q_t$\,;
		\item For all $X \in \mathcal{X}$ and $r \ge t$, $\Q_t[X_r = 0] = 1$.
	\end{enumerate}
	As a consequence, following the same line of arguments as in the discrete-time case in Theorem~\ref{thm:discrete}, there exists 
	a strictly positive, $\P$-c\`adl\`ag generalized supermartingale deflator $Y^t$ associated with $\Q_t$, such that (cf.\ \eqref{eq: construction of local supermartingale deflator}) 
	\begin{equation}\label{def: local deflator in continuous time before t}
	\forall s<t, \quad Y^t_s = 1/\widehat{f}^t_s,
	\end{equation}
	where $ \widehat{f}^t_s \in \mbox{cl}_{{\Q_t}}(\cC_s)$ is the ``static deflator'' (cf.\ Lemma~\ref{lemma: static deflator}) in  $\mbox{cl}_{{\Q_t}}(\cC_s)$ with respect to $\Q_t$, and 
	\begin{equation}\label{def: local deflator in continuous time after t}
	\forall s \ge t, \quad Y^t_s = 1.
	\end{equation}
	Furthermore, if $\infty \in \cA$,
	then by \eqref{eq: before tau in continuous time} we know that under the conditional measure $\Q_\infty$ the market $\mathcal{X}$ contains a num\'eraire, and hence by \cite[Theorem~2.3]{Kardaras.13.2} there exists with respect to $\Q_\infty$
	a strictly positive, $\Q_\infty$-c\`adl\`ag generalized supermartingale deflator $Y^\infty$. 
	Now we define $Z_s := \sum_{t \in \cA}Y^t_s\, \mathbbm{1}_{\{\tau = t\}}$ for $s \in [0,T]$. We claim that $(Z_s)_{s\in [0,T]}$ is a strictly positive, $\P$-c\`adl\`ag generalized supermartingale deflator. Indeed, since for all $t \in \cA$, $Y^t_s$ is strictly positive with respect to $\Q_t$, we  have for any $s\in [0,T]$ that
	\begin{align*}
	\P[Z_s > 0] &= \sum_{t \in \cA} \Q_t[Z_s > 0]\, \P[\tau = t] \\
	&= \sum_{t \in \cA} \Q_t[Y^t_s > 0]\, \P[\tau = t]\\
	&= \sum_{t \in \cA} \P[\tau = t] = 1,
	\end{align*}
	 Using the same argument we can also show that $Z$ is $\P$-c\`adl\`ag and that for all $r <s$ in $[0,1]$, $X \in \mathcal{X}$, and $A_s \in \cF_s$,
	\begin{equation*}
	\E_{\P}\Big[\tfrac{X_rZ_r}{X_sZ_s}\mathbbm{1}_{A_s}\Big] \le \P[A_s],
	\end{equation*}
	see also \emph{Step~2} in the proof of Theorem~\ref{thm:discrete}. This ensures that $Z$ is  
	indeed a strictly positive, $\P$-c\`adl\`ag generalized supermartingale deflator for $\mathcal{X}$ under $\P$. This finishes the proof for the non $\F$-adapted case.
	
For the case where we additionally assume that the market is $\F$-adapted, it remains to show that we can construct a strictly positive, c\`adl\`ag (and not only $\P$-c\`adl\`ag)  supermartingale deflator (i.e., $\F$-adapted one).
To that end, using the notations introduced before, 
 we can now, since the market is $\F$-adapted, apply Lemma~\ref{lemma: nice version of optimizer} to make sure that each $Y^t_s$ defined in  \eqref{def: local deflator in continuous time before t} is $\cF_s$--measurable so that $(Y^t_s)_{s\in [0,T]}$ is 
 a strictly positive, $\P$-c\`adl\`ag supermartingale deflator 
 with respect to $\Q_t$. Consequently, the process $Z_s = \sum_{t \in \cA}Y^t_s\, \mathbbm{1}_{\{\tau = t\}}$, $s\in [0,T]$, is 
 a strictly positive, $\P$-c\`adl\`ag supermartingale deflator with respect to all $\Q_t$. Consider the process $\widetilde{Z}_s := \E_{\P}[Z_s\,|\, \cF_s]$, $s \in [0,T]$. Then proceeding as in the proof of
 Theorem~\ref{thm:discrete} we can show that $X\widetilde{Z}$ satisfies the (usual) supermartingale property under $\P$ for all $X \in \mathcal{X}$. Therefore, it remains to show that $\widetilde{Z}$ admits a c\`adl\`ag and $\F$-adapted $\P$-version. To that end, let $\ov{X}\in \cX$ be the generalized num\'eraire.
   Since $\overline{X}_s\widetilde{Z}_s$ is a nonnegative $\P$-supermartingale, Fatou's lemma together with the supermartingale property guarantees the right-continuity of the map $s\mapsto \E_{\P}[\overline{X}_s\widetilde{Z}_s]$. As $\F$ satisfies the usual conditions, a classical result from probability theory (see, e.g., \cite[Theorem~VI.4, p.69]{DellacherieMeyer.82}) hence ensures that there exists a c\`adl\`ag $\P$-version of $\widetilde{Z}\overline{X}$, say, $S = (S_s)_{s \in [0,T]}$. Now we define a new strictly positive process $Z^\prime = (Z^\prime_s)_{s \in [0,T]}$ by
\begin{equation}\label{eq: regularized version of deflator}
Z^\prime_s := \tfrac{S_s}{\overline{X}_s}\,\mathbbm{1}_{\{\overline{X}_s > 0\}} + \mathbbm{1}_{\{\overline{X}_s = 0\}}, \quad s \in [0,T].
\end{equation}
In view of the property (ii) of a wealth process in Definition~\ref{def:wealth-set}, $\overline{X}_s(\omega) = 0$ implies that $\overline{X}_r(\omega) = 0$ for all $r \ge s$, which in turn implies the c\`adl\`ag property of $s \mapsto \mathbbm{1}_{\{\overline{X}_s = 0\}}$. Therefore, since $S$ and $\overline{X}$ all have c\`adl\`ag paths, also  $Z^\prime$ has c\`adl\`ag paths. Finally, for any $X \in \mathcal{X}$, as $\{X_t > 0\} \subseteq \{\overline{X}_t > 0\}$ holds $\P$--a.s., we have $\P$--a.s. that
\begin{equation*}
X_tZ^\prime_t = X_t \frac{S_t}{\overline{X}_t}\,\mathbbm{1}_{\{\overline{X}_t > 0\}} = X_t\widetilde{Z}_t\,\mathbbm{1}_{\{X_t > 0\}} = X_t\hat{Z}_t.
\end{equation*}
This together with the supermartingale deflator property of $\widetilde{Z}$ shows that $Z^\prime$ is indeed 
a strictly positive, c\`adl\`ag supermartingale deflator for $\mathcal{X}$.
\end{proof}
\subsubsection{Proof of Theorem~\ref{theorem: third main theorem}}\label{subsubsec:third-main-thm}
In this subsubsection we provide the proof of Theorem~\ref{theorem: third main theorem}. 
To that end, let us first argue why without loss of generality we may assume in its proof that both
\begin{equation}\label{eq: assumption, no atoms}
\mathcal{A} = \emptyset
\end{equation}
and $\mathcal{J} = \text{supp}(\cL(\tau)) = [0,T]$, or equivalently,
\begin{equation}\label{eq: assumption, support is the whole interval}
\forall \text{ open interval } J \subseteq [0,T], \quad \cL(\tau)[J] > 0.
\end{equation}
 Indeed, recall that by assumption $\mathcal{J}= \text{supp}(\cL(\tau)) \setminus \cA$ is not empty. Moreover, from the proof of Theorem~\ref{theorem: second main theorem} we have seen that for each $t \in \cA$ one can find a 
 strictly positive, ($\P$-)c\`adl\`ag
 (generalized) supermartingale deflator $Y^t$ for $\mathcal{X}$ with respect to $\P[\,\cdot \,|\,\tau = t]$. Then by pasting them together we can obtain a process $Z^{\cA}$  defined by
\begin{equation*}
Z^{\cA}_s := \sum_{t \in \cA}Y^t_s\, \mathbbm{1}_{\{\tau = t\}}, \quad s \in [0,T],
\end{equation*}
such that $Z^{\cA}$ it is 
a strictly positive, ($\P$-)c\`adl\`ag
(generalized) supermartingale deflator for $\mathcal{X}$ under the conditional measure $\P[\,\cdot\,|\,\tau \in \cA]$. Now, denote by $\mathcal{J}_1,\dots,\mathcal{J}_n$ the connected component of $\mathcal{J}$ . If we can find for each $\mathcal{J}_i$, $i:=1,\dots,n$, 
a strictly positive, ($\P$-)c\`adl\`ag
(generalized) supermartingale deflator $Z^{\mathcal{J}_i}$ 
 under the conditional measure $\P[\,\cdot\,|\, \tau \in \mathcal{J}_i]$, then by the same pasting arguments as in the discrete case the process
 \begin{equation*}
 Z_s:= Z^{\cA}_s\, \mathbbm{1}_{\cA}(s) + \sum_{i=1}^n Z^{\mathcal{J}_i}_s\,\mathbbm{1}_{\mathcal{J}_i}(s), \quad s\in [0,T],
 \end{equation*}
   will be 
   a strictly positive, ($\P$-)c\`adl\`ag
   (generalized) supermartingale deflator with respect to $\P$. Hence, in order to keep the notation short, we may indeed without loss of generality assume for the rest of this subsubsection that \eqref{eq: assumption, no atoms} and \eqref {eq: assumption, support is the whole interval} hold.

Let us first prove the simpler direction $(ii) \Rightarrow (i)$ in Theorem~\ref{theorem: third main theorem}.
\begin{proof}[Proof of $(ii) \Rightarrow (i)$ in Theorem~\ref{theorem: third main theorem} with \eqref{eq: assumption, no atoms} and \eqref{eq: assumption, support is the whole interval}]
	
	By the same argument as, e.g., in the proof of Theorem~\ref{thm:discrete}, we see that  the existence of a strictly positive (generalized) supermartingale deflator $Z$
	implies that NUPBR$_s$ holds for each $s \in [0,T]$.
	
	Now to see that also $(i)(b)$ holds, note that by assumption there exist a countable dense subset $\cD \subseteq [0,T]$ containing $0$ and $T$ and a strictly positive process $Z^\infty$ defined on $\cD$ with $Z^\infty_0 \le 1$ such that  for all $r<u \in \cD$ 
	\begin{equation*}
	\E_{\Q_{r,u}}\Big[\frac{X_tZ^\infty_t}{X_0Z^\infty_0}\Big] \le 1
	\end{equation*}
	for all $X \in\mathcal{X}$ under $\Q_{r,u}$, where we recall that $\Q_{r,u}[\,\cdot\,] = \P\big[\,\cdot\,|\,\tau \in (r,u]\big]$, see Notation~\ref{not:J-A-cond-prob}. Then, since we have $X_0 = 1$, see Definition~\ref{def:wealth-set},  the above inequality implies that $\E_{\Q_{r,u}}[X_tZ^\infty_t] \le 1$. Hence, for any real number $M > 0$, the Markov inequality implies that $\Q_{r,u}[X_tZ^\infty_t \ge M] \le 1/M$. Consequently, we have for all $X \in \mathcal{X}$, $M >0$,
	\begin{align*}
	\Q_{r,u}\big[X_t \ge M\big] 
	&\le 
	\Q_{r,u}\big[X_tZ^\infty_t \ge \sqrt{M}\big] + \Q_{r,u}\big[\tfrac{1}{Z^\infty_t} \ge \sqrt{M}\big] \\
	&\le
	 \tfrac{1}{\sqrt{M}} + \Q_{r,u}\big[\tfrac{1}{Z^\infty_t} \ge \sqrt{M}\big].
	\end{align*}
	Now, invoking \eqref{eq: uniform boundedness of deflator}, for any given $\varepsilon > 0$ we can find $M$ large enough such that $\Q_{r,u}[\tfrac{1}{Z^\infty_t} \ge \sqrt{M}] < \varepsilon$ holds uniformly over all $X \in \mathcal{X}$ and $u>r\ge t$ in $\cD$. Hence, by combining all above estimates we can derive that
	\begin{equation*}
 \limsup_{M \rightarrow \infty}\sup_{r,u \in \mathcal{D}, u >r\ge t}\sup_{X_t \in \mathcal{C}_t}\Q_{r,u}\big[X_t \ge M\big] = 0,
	\end{equation*}
	which is exactly the desired equation~\eqref{eq: uniform boundedness in probability}. This completes the proof of $(ii) \Rightarrow (i)$ in Theorem~\ref{theorem: third main theorem}.
\end{proof}
The most technical part of this paper is to prove the implication $(i) \Rightarrow (ii)$ in Theorem~\ref{theorem: third main theorem}. For the ease of notation, we will  without loss of generality assume that the set $\cD$ in $(i)$ of Theorem~\ref{theorem: third main theorem} satisfies the following:
\begin{equation*}
\mbox{$\cD$ is the set of all dyadic numbers in $[0,T]$}.
\end{equation*}
Moreover, we denote by $\cD_k:=\{\nicefrac{iT}{2^k}\colon i:=0,1,\dots,2^k\}$ the collection of all $k$-th dyadic numbers in $[0,T]$. 

 Indeed, a careful look through the proof shows that the only property we actually use from the  set of dyadic numbers $\cD$ is that it is dense and that there exists an increasing sequence of finite sets $\cD_0\subseteq \cD_1\subseteq\cD_2\subseteq\dots$ with $\{0,T\}\subseteq \cD_i$ for each $i$ satisfying  $\cD = \bigcup_{k \in \N} \cD_k$, which of course can be constructed for any countable dense subset $\cD$ which contains $0$ and $T$.

 Then, for every fixed $k \in \N$ and $r  \in \cD_k\setminus \{T\}$, we use $\Q^k_r[\,\cdot\,]$ to denote the conditional probability
\begin{equation}\label{def: conditional measure in continuous time}
\Q^k_r[\,\cdot\,] := \P\big[\,\cdot\,\big|\, \tau \in (r, r+\tfrac{T}{2^k}]\big].
\end{equation}
Moreover, a crucial object in the proof will be for each $t \in \cD$ the set
\begin{equation}\label{eq:B-t}
B_t := \{k \in \N: t \in \cD_k\}.
\end{equation}
 Then, note  under Assumptions~\eqref{eq: assumption, no atoms} and \eqref{eq: assumption, support is the whole interval}, Equation~\eqref{eq: uniform boundedness in probability} implies that for all $t \in \cD$, 
\begin{equation}\label{eq: (1.2) in dyadic case}
\limsup_{M \rightarrow \infty}\sup_{k \in B_t}\sup_{r \in \cD_k, r\ge t}\sup_{X_t \in \cC_t}\Q^k_r[X_t \ge M] = 0.
\end{equation}
With these preparations,  we are now able to prove the implications $(i) \Rightarrow (ii)$ in Theorem~\ref{theorem: third main theorem}.
Due to its technicality and length, we will divide  the proof of the implication $(i) \Rightarrow (ii)$ in Theorem~\ref{theorem: third main theorem} into five steps.
\begin{proof}[Proof of $(i) \Rightarrow (ii)$ in Theorem~\ref{theorem: third main theorem} with \eqref{eq: assumption, no atoms} and \eqref{eq: assumption, support is the whole interval}]
	First we note that the assumption~\eqref{eq: assumption, no atoms} implies that $\P[\tau = \infty] = 0$, which means that $\P[\overline{X}_T = 0] = 1$. This in turn implies that $\P[X_T=0]=1$ for all $X\in \cX$, meaning that the market has died out at the terminal time $T$. Hence, it suffices to consider the market on the time interval $[0,T)$.

	\textit{Step~1}: (Local deflators on $\cD_k$).
	For each $k$, denote by $\mathcal{X}^k:=\{(X_t)_{t \in \cD_k}: X \in \mathcal{X}\}$ the restriction of the market to the $k$-dyadic grid.
	To show that there exists for each $k$ a deflator $(Z^k_s)_{s\in \cD_k}$ on $\cX^k$ with respect to each $\Q^k_r$, $r\in \cD_k \setminus\{T\}$, we follow the proof of Theorem~\ref{thm:discrete}. 
	For a fixed $k \in \cD_k$ and an $r \in \cD_k \setminus\{T\}$, we see from the definition of $\tau$ in \eqref{eq: tau in continuous time} and the definition of the conditional measure $\Q^k_r$ in \eqref{def: conditional measure in continuous time} that
	\begin{enumerate}
		\item The market $\mathcal{X}$ contains a num\'eraire on $[\![0,r]\!]$ under the conditional measure $\Q^k_r$;
		\item For all $X \in \mathcal{X}$ and $u \ge r+\frac{T}{2^k}$, $\Q^k_r[X_u = 0] = 1$.
	\end{enumerate}
	Since by assumption $(i)$ NUPBR$_r$ holds for all $r$, we can obtain as in the discrete-time case (see Theorem~\ref{thm:discrete}) 
	a 
	strictly positive
	generalized supermartingale deflator $Y^{(k,r)}_s$, $s \in \cD_k$, for $\mathcal{X}^k$ with respect to $\Q^k_r$. In particular, invoking the explicit construction of such $Y^{(k,r)}$, see \eqref{eq: construction of local supermartingale deflator} or \eqref{def: local deflator in continuous time before t} and \eqref{def: local deflator in continuous time after t}, we know that $(Y^{(k,r)}_s)_{s \in \cD_k}$, can be formulated as
	\begin{equation}\label{eq: local deflator for kth dyadic points in continuous time} 
	Y^{(k,r)}_s = \frac{1}{\widehat{f}^{(k,r)}_s}\,\mathbbm{1}_{\{s\le r\}} +  \mathbbm{1}_{\{s>r\}}, \quad s \in \cD_k,
	\end{equation}
	where $\widehat{f}^{(k,r)}_s \in \mbox{cl}_{\Q^k_r}(\cC_s)$ is the \textit{static deflator}  in $\mbox{cl}_{\Q^k_r}(\cC_s)$ with respect to $\Q^k_r$ (cf.\ Lemma~\ref{lemma: static deflator}).
	Now we can paste these ``local deflators'' over all $r \in \cD_k \setminus \{T\}$ as in the discrete case before. More precisely, using \eqref{eq: local deflator for kth dyadic points in continuous time}, we define the process $(Z^k_s)_{s\in \cD_k}$ by setting for each $s\in \cD_k$ 
	\begin{equation}\label{eq: P deflators for kth dyadic points in continuous time}
	Z^k_s 
	:=
	 \sum_{r \in \cD_k, r<T}Y^{(k,r)}_s\, \mathbbm{1}_{\{\tau \in (r, r+ \frac{T}{2^k}]\}} 
	= 
	\sum_{r \in \cD_k, s\le r<T} \frac{1}{\widehat{f}^{(k,r)}_s}\,\mathbbm{1}_{\{\tau \in (r, r+ \frac{T}{2^k}]\}} + \mathbbm{1}_{\{\tau \le s\}}.
	\end{equation}
	Following the arguments of Theorem~\ref{thm:discrete} and Remark~\ref{remark: global deflator is also local deflator}, we conclude  that $(Z^k_s)_{s\in \cD_k}$
 is a strictly positive
 generalized supermartingale deflator for $\cX^k$ with respect to $\P$ and all $\Q^k_r$, $r \in \cD_k \setminus\{T\}$. This  means for all $r \in \cD_k \setminus\{T\}$, $s<t\in\cD_k$, $A_s \in \cF_s$, it holds that
	\begin{equation}\label{eq: supermartingale deflator property globally and locally on Dk}
	\E_{\P}\big[\tfrac{X_tZ^k_t}{X_sZ^k_s}\,\mathbbm{1}_{A_s}\big] \le \P[A_s], \quad \E_{\Q^k_r}\Big[\tfrac{X_tZ^k_t}{X_sZ^k_s}\,\mathbbm{1}_{A_s}\Big] \le \Q^k_r(A_s),
	\end{equation}
	which finishes \textit{Step~1} of the proof.

\textit{Step~2}: (Extension from $\cD_k$ to $\cD$).
In  \textit{Step~1}, we have obtained a sequence $(Z^k)_{k \in \N}$ such that for each $k$, $Z^k$ is a strictly positive 
generalized supermartingale deflator for $\mathcal{X}^k$ under $\P$ and under all $\Q^k_r$, $r \in \cD_k \setminus\{T\}$. Like for the discrete-case we did not need the uniform boundedness assumption~\eqref{eq: uniform boundedness in probability} in the construction of these $Z^k$.
In \textit{Step~2}, we want to construct a process $Z^{\cD}$ defined on $\cD$ based on $(Z^k)_{k \in \N}$ such that $Z^{\cD}$ satisfies certain strictly positive 
generalized supermartingale deflator properties for $\mathcal{X}^{\cD}:= \{(X_t)_{t \in \cD}: X \in \mathcal{X}\}$ with respect to $\P$ and all $\Q^k_r$, $r \in \cD_k \setminus\{T\}$, $k \in \N$, cf. Remark~\ref{rem:main-thm-meaning}. To achieve this goal, the assumption~\eqref{eq: uniform boundedness in probability} is crucial.
To that end, let us first start with some simple observations. Denote for each $k \in \N$ and $r \in \cD_k \setminus \{T\}$ the set
\begin{equation*}
A^k_r:= \big\{\tau \in (r, r+\tfrac{T}{2^k}]\big\}.
\end{equation*}
Now, let $t \in \cD$ and consider $k \in B_t$, where $B_t$ is defined in \eqref{eq:B-t}. For every $r\geq t$ these $A^k_r$, $r \in \cD_k \setminus \{T\}$, are pairwise disjoint, and we also have that $\bigcup_{r \ge t,r \in \cD_k\setminus\{T\}}A^k_r = \{\tau >t\}$ is disjoint from the set $\{\tau \le t\}$. Moreover, for $r \in \cD_k \setminus \{T\}$,  we have that
\begin{equation}\label{eq: a trivial union relation}
A^{k+1}_r \cup A^{k+1}_{r + \frac{T}{2^{k+1}}} = A^k_r.
\end{equation}

\textit{Step~2.1}: (Boundedness of convex combinations of $1/Z^k$).
For each $t \in \cD$, we claim that  the condition~\eqref{eq: uniform boundedness in probability} (or more precisely condition \eqref{eq: (1.2) in dyadic case}) implies that the convex hull of $\frac{1}{Z^k_t}$, $k \in B_t$, denoted by
\begin{equation}\label{eq: forward convex hull of deflators}
K_t := \text{conv}\Big\{\tfrac{1}{Z^k_t}: k \in B_t\Big\}
\end{equation}
is $\P$-bounded. Note that since every $(Z^k)$ is a strictly positive generalized supermartingale deflator for $\cX^k$ under $\P$ and $t\in \cD_k$, we have by definition that $Z^k_t > 0$ holds $\P$--a.s.; in particular, $\frac{1}{Z^k_t}$ is well defined.

Now, in view of \eqref{eq: P deflators for kth dyadic points in continuous time} and the disjointness relation between $A^k_r$, $r\geq t$, and $\{\tau \le t\}$, we  see that
\begin{equation}\label{eq: formula of reciprocal of Zk}
\frac{1}{Z^k_t} = \sum_{r \in \cD_k, t\le r<T} \widehat{f}^{(k,r)}_t\,\mathbbm{1}_{A^k_r} + \mathbbm{1}_{\{\tau \le t\}},
\end{equation}
	where $\widehat{f}^{(k,r)}_t \in \mbox{cl}_{\Q^k_r}(\cC_t)$ is the \textit{static deflator}  in $\mbox{cl}_{\Q^k_r}(\cC_t)$ with respect to $\Q^k_r$ (cf.\ Lemma~\ref{lemma: static deflator}). 
Now we fix an $\varepsilon>0$ and consider $\lambda \frac{1}{Z^k_t} + (1-\lambda)\frac{1}{Z^{k+1}_t}$ for any $\lambda \in [0,1]$. By \eqref{eq: formula of reciprocal of Zk} and \eqref{eq: a trivial union relation} we have
\begin{equation}\label{eq:Step2.1-1}
\begin{split}
&\lambda \frac{1}{Z^k_t} + (1-\lambda)\frac{1}{Z^{k+1}_t}\\
&= 
\sum_{r \in \cD_k, t\le r<T}\Big(\lambda \widehat{f}^{(k,r)}_t + (1-\lambda)\widehat{f}^{(k+1,r)}_t\Big)\,\mathbbm{1}_{A^{k+1}_r}
 \\
& \quad + \sum_{r \in \cD_k, t\le r<T}\Big(\lambda \widehat{f}^{(k,r)}_t + (1-\lambda)\widehat{f}^{(k+1,r+\frac{T}{2^{k+1}})}_t\Big)\,\mathbbm{1}_{A^{k+1}_{r+\frac{T}{2^{k+1}}}} + \mathbbm{1}_{\{\tau \le t\}}.
\end{split}
\end{equation}
Moreover, since $\Q^{k+1}_r \ll \Q^k_r$ and $\Q^{k+1}_{r + \frac{T}{2^{k+1}}}\ll \Q^k_r$ holds for all $ r \in \cD_k\setminus\{T\}$, and $\cC_t$ is convex, we have 
\begin{equation*}
\begin{split}
&\lambda \widehat{f}^{(k,r)}_t + (1-\lambda)\widehat{f}^{(k+1,r)}_t \in \mbox{cl}_{\Q^{k+1}_r}({\cC_t}),
\\
&\lambda \widehat{f}^{(k,r)}_t + (1-\lambda)\widehat{f}^{(k+1,r+\frac{T}{2^{k+1}})}_t \in \mbox{cl}_{\Q^{k+1}_{r + \frac{T}{2^{k+1}}}}({\cC_t}),
\end{split}
\end{equation*}
 which together with Lemma~\ref{lemma: closure and closure of solid hull is bounded} allows us to use the uniform boundedness condition~\eqref{eq: (1.2) in dyadic case} to find an $M>1$ independent of $k \in B_t$, $r \in \cD_k$, and $\lambda \in [0,1]$ such that both
\begin{equation}\label{eq:Step2.1-M}
\begin{split}
\Q^{k+1}_r\Big[\lambda \widehat{f}^{(k,r)}_t + (1-\lambda)\widehat{f}^{(k+1,r)}_t\ge M\Big] &\le \varepsilon,  
\\
\Q^{k+1}_{r+\frac{T}{2^{k+1}}}\Big[\lambda \widehat{f}^{(k,r)}_t + (1-\lambda)\widehat{f}^{(k+1,r+\frac{T}{2^{k+1}})}_t\ge M\Big] &\le \varepsilon.
\end{split}
\end{equation}
Now, denote 
\begin{equation}\label{eq:def-h-k}
\begin{split}
h^{(k+1,r)}_t&:=\lambda \widehat{f}^{(k,r)}_t + (1-\lambda)\widehat{f}^{(k+1,r)}_t,\\ 
h^{(k+1,r+ \frac{T}{2^{k+1}})}_t&:= \lambda \widehat{f}^{(k,r)}_t + (1-\lambda)\widehat{f}^{(k+1,r+\frac{T}{2^{k+1}})}_t.
\end{split}
\end{equation}
Then \eqref{eq:Step2.1-1} and \eqref{eq:Step2.1-M} together with the observation that on $\{\tau\leq t\}$ one has $Z^k_t=Z^{k+1}_t=1$, ensures for such $M>1$   that for all $k \in B_t$
\begin{align*}
\P\big[\lambda\tfrac{1}{Z^k_t} + (1-\lambda)\tfrac{1}{Z^{k+1}_t} \ge M\big]
&= \P\big[\{\lambda\tfrac{1}{Z^k_t} + (1-\lambda)\tfrac{1}{Z^{k+1}_t} \ge M\}\cap \{\tau>t\}\big]\\
& \le \sum_{r \in \cD_k, t\le r<T} \Q^{k+1}_r\big[h^{(k+1,r)}_t\ge M\big] \,\P\big[A^{k+1}_r\big] 
\\&
\quad 
+ \sum_{r \in \cD_k, t\le r<T} \Q^{k+1}_{r + \frac{T}{2^{k+1}}}\big[h^{(k+1,r+ \frac{T}{2^{k+1}})}_t\ge M\big]\, \P\big[A^{k+1}_{r + \frac{T}{2^{k+1}}}\big] \\
&\le \varepsilon  \sum_{r \in \cD_k, t\le r<T}\Big(\P\big[A^{k+1}_r\big] + \P\big[A^{k+1}_{r + \frac{T}{2^{k+1}}}\big]\Big) \\
&\le \varepsilon,
\end{align*}
where for the last inequality we used that 
\begin{equation*}
\sum_{r \in \cD_k, t\le r<T}\Big(\P\big[A^{k+1}_r\big] + \P\big[A^{k+1}_{r + \frac{T}{2^{k+1}}}\big]\Big)
 = 
 \P\bigg[\bigcup_{r \ge t,r \in \cD_k\setminus\{T\}}(A^{k+1}_r \cup A^{k+1}_{r + \frac{T}{2^{k+1}}} )\bigg] \le 1.
\end{equation*}
Since $M>1$ was independent of $k \in B_t$, $r \in \cD_k$, and $\lambda \in [0,1]$, thanks to the uniform boundedness condition~\eqref{eq: (1.2) in dyadic case}, we can show the same result for any convex combination of $\frac{1}{Z^k_t}$ for $k \in B_t$, which gives the $\P$--boundedness of the convex hull $K_t$ defined in \eqref{eq: forward convex hull of deflators}.
	
\textit{Step~2.2}: (An application of Komlos lemma).
Let  $t \in \cD$ and consider the sequence of nonnegative random variable $(Z^k_t)_{k \in B_t}$, where $Z^k$ is the strictly positive 
 generalized supermartingale deflator for $\mathcal{X}^k$ constructed in \textit{Step~1}. We first claim that that the convex set $\text{conv}\big\{Z^k_t: k \in B_t\big\}$ is $\P$-bounded.
Indeed, for all $k \in B_t$, using the generalized supermartingale property (with respect to $\P$), see \eqref{eq: supermartingale deflator property globally and locally on Dk}, we have
\begin{equation*}
\E_{\P}\big[\overline{X}_tZ^k_t\big] \le 1.
\end{equation*}
Moreover, when considering any convex combinations $Y_t:= \lambda_1Z^{k_1}_t + \ldots + \lambda_lZ^{k_l}_t$   of $Z^k_t$'s for $k_1,\dots,k_l \in B_t$, we also have  
that
\begin{equation}\label{eq:2.2-1}
\E_{\P}\big[\overline{X}_t\,Y_t\big] \le 1.
\end{equation}
This in turn implies that the convex hull $\text{conv}\big\{Z^k_t: k \in B_t\big\}$ of all $(Z^k_t)_{k \in B_t}$ is $\P$-bounded on $\{\overline{X}_t > 0\}$. 
On the other hand, from \eqref{eq: P deflators for kth dyadic points in continuous time} we see that $Z^k_t = 1$ on $\{\tau \le t\}$ for all $k \in B_t$. Moreover, $\{\overline{X}_t = 0\} \subseteq \{\tau \le t\}$  holds as $\tau$ is the first hitting time of $\overline{X}$ at $0$.  Therefore, we   have $Z^k_t = 1$ on  $\{\overline{X}_t = 0\}$ for all $k \in B_t$ which ensures that also  $Y_t = 1$ on $\{\overline{X}_t = 0\}$  for any convex combination $Y_t:= \lambda_1Z^{k_1}_t + \ldots + \lambda_lZ^{k_l}_t$   of $Z^k_t$'s for $k_1,\dots,k_l \in B_t$. This means  that $\text{conv}\big\{Z^k_t: k \in B_t\big\}=1$ on $\{\overline{X}_t = 0\}$, in particular it is $\P$-bounded also on $\{\overline{X}_t = 0\}$. Combining the two facts above together we can conclude that the convex set $\text{conv}\big\{Z^k_t: k \in B_t\big\}$ is indeed $\P$-bounded.
Now we can apply Komlos lemma for nonnegative random variables, see \cite[Appendix]{Kabanov97}, to the sequence $(Z^k_t)_{k \in B_t}$ to get a sequence of forward convex combinations of $Z^k_t$, $k \in B_t$, which is denoted by $\text{\textbf{fconv}}Z^k_t$, $k \in B_t$, such that $(\text{\textbf{fconv}}Z^k_t)_{k \in B_t}$ converges to a nonnegative random variable $Z^\infty_t$ $\P$--a.s.\ as $k \rightarrow \infty$. The boundedness of $\text{conv}\big\{Z^k_t: k \in B_t\big\}$ then guarantees that $Z^\infty_t < \infty$ $\P$--a.s.

Moreover, we just showed in \textit{Step~2.1} that the convex hull $K_t$ of reciprocals of all $Z^k_t$, see \eqref{eq: forward convex hull of deflators}, is also $\P$-bounded. Then, since the function $(0,\infty)\ni x \mapsto 1/x$ is convex, we obtain  that
\begin{equation*}
\frac{1}{\lambda_1 Z^k_t + \ldots +\lambda_{l}Z^{k+l-1}_t} \le \lambda_1 \frac{1}{Z^k_t} + \ldots + \lambda_l\frac{1}{Z^{k+l-1}_t}
\end{equation*}
for any convex weights $\lambda_j$, $j=1,\ldots,l$, and consequently the sequence
$\big(\frac{1}{\text{\textbf{fconv}}Z^k_t}\big)_{k \in B_t}$ is also $\P$-bounded. Then since $\frac{1}{\text{\textbf{fconv}}Z^k_t}$ converges to $\frac{1}{Z^\infty_t}$ $\P$--a.s., we have $\frac{1}{Z^\infty_t} < \infty$, or equivalently, $Z^\infty_t > 0$ $\P$--a.s.

We  point out that 
in the above convergent sequence $(\text{\textbf{fconv}}Z^k_t)_{k \in B_t}$, where 
\begin{equation*}
\text{\textbf{fconv}}Z^k_t = \lambda^k_1 Z^k_t + \ldots +\lambda^k_{l}Z^{k+l-1}_t,
\end{equation*}
these convex weights $(\lambda^k_j)_{j:=1,\dots l}$ and the indices $l \in \N$, may depend on $t \in \cD$, and hence may vary when doing the above procedure 
separately for each $t \in \cD$.
However, it turns out in the later part of the proof that we would like to have for each $s,t \in \cD$ joint convex weights (and indices) such that both sequences of forward convex combintations with respect to $(Z^k_s)$ and $(Z^k_t)$ converge.

To see this, we first  consider $\cD_0 := \{0,T\}$. 
By Komlos lemma, we can find a sequence of forward convex combinations of $Z^k_0$, $k \in \N$, which now is denoted by $(\text{\textbf{fconv}}_0Z^k_0)_{k \in \N}$, such that it converges to a strictly positive, finite valued random variable $Z^\infty_0$ $\P$--a.s.. Then we apply Komlos lemma to the sequence $(\text{\textbf{fconv}}_0Z^k_T)_{k \in \N}$, where $(\text{\textbf{fconv}}_0 Z^k_T)_{k \in \N}$ possesses the same forward convex combination form as in the sequence $(\text{\textbf{fconv}}_0Z^k_0)_{k \in \N}$, but with $Z^k_0$ replaced by $Z^k_T$, to get a convergent sequence $(\text{\textbf{fconv}}_TZ^k_T)_{k \in \N}$ which consists of forward convex combinations of $(\text{\textbf{fconv}}_0Z^k_T)_{k \in \N}$, and converges to a strictly positive, finite valued random variable $Z^\infty_T$ $\P$--a.s.. Note that when using the convex combinations appeared in $(\text{\textbf{fconv}}_TZ^k_T)_{k \in \N}$, the sequence $(\text{\textbf{fconv}}_TZ^k_0)_{k \in \N}$ still converges to $Z^\infty_0$ $\P$--a.s.

Suppose now that for some $n \in \N$ we have already found a forward convex combinations form $\text{\textbf{fconv}}_{\frac{(2^n-1)T}{2^n}}$ such that for \textit{all} $t \in \cD_n$, the sequence $(\text{\textbf{fconv}}_{\frac{(2^n-1)T}{2^n}}Z^k_t)_{k \in B_t}$ converges to a strictly positive, finite valued random variable $Z^\infty_t$ $\P$--a.s.\ as $k \rightarrow \infty$. Then we apply Komlos lemma to the sequence $(\text{\textbf{fconv}}_{\frac{(2^n-1)T}{2^n}}Z^k_{\frac{T}{2^{n+1}}})_{k \in B_{\frac{T}{2^{n+1}}}}$ to obtain a forward convex combination subsequence $(\text{\textbf{fconv}}_{\frac{T}{2^{n+1}}}Z^k_{\frac{T}{2^{n+1}}})_{k \in B_{\frac{T}{2^{n+1}}}}$ which is convergent to a strictly positive, finite valued random variable $Z^\infty_{\frac{T}{2^{n+1}}}$ $\P$--a.s.\ as $k \rightarrow \infty$. Notice that the sequence $(\text{\textbf{fconv}}_{\frac{T}{2^{n+1}}}Z^k_t)_{k \in B_t}$ still converges to $Z^\infty_t$ $\P$--a.s. (as $k \rightarrow \infty$) simultaneously for all $t \in \cD_n \cup \{\frac{T}{2^{n+1}}\}$. Repeating this argument for all $t \in \cD_{n+1} \setminus \cD_n$ we can get a forward convex combination form $\text{\textbf{fconv}}_{\frac{(2^{n+1}-1)T}{2^{n+1}}}$ such that  $(\text{\textbf{fconv}}_{\frac{(2^{n+1}-1)T}{2^{n+1}}}Z^k_t)_{k \in B_t}$ converges to $Z^\infty_t$ $\P$--a.s.\ (as $k \rightarrow \infty$) simultaneously for all $t \in \cD_{n+1}$. Hence we can complete this induction proof to see that our claim holds. More precisely, when denoting for any $s,t \in \cD$ 
\begin{equation*}
\begin{split}
m_{s,t}&:= \min\{j\in\N: s \in \cD_j\} \vee \min\{j\in\N: t \in \cD_j\},\\
p(s,t) &:= \frac{(2^{m_{s,t}}-1)T}{2^{m_{s,t}}},
\end{split}
\end{equation*}
 we indeed obtained by the above procedure the desired property that 
\begin{align}\label{eq: cantor diagonal type convergence}
\forall s,t \in \cD,  \quad 
&\lim_{k \rightarrow \infty} (\text{\textbf{fconv}}_{p(s,t)}Z^k_t) = Z^\infty_t 
\quad \mbox{ \textit{and} } \quad \lim_{k \rightarrow \infty} (\text{\textbf{fconv}}_{p(s,t)}Z^k_s) = Z^\infty_s.
\end{align}
We finish \textit{Step~2} by remarking that as each $Z^k$ is a  generalized supermartingale deflator for $\cX^k$ with respect to $\P$, we have by definition that $Z^k_0 \le 1$ for all $k \in \N$ . This, in turn, ensures that
 $Z^\infty_0 \le 1$  holds as well.

	\textit{Step~3:} (
	Generalized supermartingale property (GSP) on $\cD$).
	
	\textit{Step~3.1:} (GSP on $\cD$ for $\P$--expectations).
Our goal in \textit{Step~3.1} is to show that 	
 for all $s<t$ in $\cD$ and for all $X \in \mathcal{X}$,
\begin{equation}\label{eq: P-expectation of discounted processes}
\E_{\P}\Big[\tfrac{X_tZ^\infty_t}{X_sZ^\infty_s}\Big] \le 1.
\end{equation}
	To that end, let $t>s$ in $\cD$, $k \in B_t \cap B_s$ (which means that $s,t \in \cD_k$) be fixed. We first claim that for any $X \in \mathcal{X}$ and $l \ge 1$, 
	\begin{equation}\label{eq: estiamtes for P-expectation}
	\E_{\P}\Big[\tfrac{X_tZ^{k+l}_t}{X_sZ^k_s}\Big] \le 1.
	\end{equation}
	We first consider the case $l = 1$. Note that on the event $\{\tau \le t\}$ we have $X_t = 0$ and hence  by Remark~\ref{remark: well defined condition}, it remains to show that
	\begin{equation}\label{eq: estiamtes for P-expectation-red}
	\E_{\P}\Big[\tfrac{X_tZ^{k+l}_t}{X_sZ^k_s}\,\mathbbm{1}_{\{\tau > t\}}\Big] \le 1.
	\end{equation}
	To that end, in view of the formula \eqref{eq: formula of reciprocal of Zk} for $Z^k$ and $Z^{k+1}$ we have
	\begin{equation}\label{eq: quotient between k+1 and k}
	\frac{X_tZ^{k+1}_t}{X_sZ^k_s}\,\mathbbm{1}_{\{\tau > t\}}
	= 
	\sum_{r \in \cD_k, t\le r<T} \bigg(\tfrac{X_t/\widehat{f}^{(k+1,r)}_t}{X_s/\widehat{f}^{(k,r)}_s}\mathbbm{1}_{A^{k+1}_r} + \tfrac{X_t/\widehat{f}^{(k+1,r+\frac{T}{2^{k+1}})}_t}{X_s/\widehat{f}^{(k,r)}_s}\mathbbm{1}_{A^{k+1}_{r + \frac{T}{2^{k+1}}}}\bigg),
	\end{equation}
	where $\widehat{f}^k_t$, $\widehat{f}^{k+1}_t$ and $\widehat{f}^{(k+1,r+\frac{T}{2^{k+1}})}_t$ are the corresponding \textit{static deflator} for the closure of $\cC_t$ under the measure $\Q^k_r$, $\Q^{k+1}_r$, and $\Q^{k+1}_{r+\frac{T}{2^{k+1}}}$, respectively; the same holds when replacing $t$ by $s$.
	Moreover, by applying Lemma~\ref{lemma: nice version of optimizer}, we choose  each $\widehat{f}^{(j,q)}_t$, $\widehat{f}^{(j,q)}_s$, for $j = k, k+1$, $q = r, r+ \frac{T}{2^{k+1}}$, above to be $\cF_t$ and $\cF_s$--measurable, respectively.
	
	Now, for a fixed $r \in \cD_k$ with $r \ge t$ 
	we observe 
	that
	\begin{align*}
	\E_{\P}\Big[ \tfrac{X_t/\widehat{f}^{(k+1,r)}_t}{X_s/\widehat{f}^{(k,r)}_s}\mathbbm{1}_{A^{k+1}_r} \Big] 
	&= 
	\E_{\Q^{k+1}_r}\Big[\tfrac{X_t/\widehat{f}^{(k+1,r)}_t}{X_s/\widehat{f}^{(k,r)}_s}\Big]\P[A^{k+1}_r] \\
&= \E_{\Q^{k+1}_r}\Big[\tfrac{X_t/\widehat{f}^{(k+1,r)}_t}{X_s/\widehat{f}^{(k+1,r)}_s}\tfrac{\widehat{f}^{(k,r)}_s}{\widehat{f}^{(k+1,r)}_s}\Big]\,\P[A^{k+1}_r]
	\end{align*}
	Recall that in \eqref{eq: supermartingale deflator property globally and locally on Dk} we have shown that $X/\widehat{f}^{(k+1,r)}$ is a generalized supermartingale with respect to $\Q^{k+1}_r$, i.e.,
	\begin{equation}\label{eq:3.1-3.2}
	 \E_{\Q^{k+1}_r}\Big[\tfrac{X_t/\widehat{f}^{(k+1,r)}_t}{X_s/\widehat{f}^{(k+1,r)}_s}\Big|\cF_s\Big] \le 1.
	 \end{equation} So, as $\frac{\widehat{f}^{(k,r)}_s}{\widehat{f}^{(k+1,r)}_s}$ is $\cF_s$--measurable, we have
	\begin{equation}\label{eq:here-meas-necessary!}
	\E_{\Q^{k+1}_r}\Big[\tfrac{X_t/\widehat{f}^{(k+1,r)}_t}{X_s/\widehat{f}^{(k+1,r)}_s}\tfrac{\widehat{f}^{(k,r)}_s}{\widehat{f}^{(k+1,r)}_s}\Big] \le \E_{\Q^{k+1}_r}\Big[\tfrac{\widehat{f}^{(k,r)}_s}{\widehat{f}^{(k+1,r)}_s}\Big].
	\end{equation}
	Also by recalling  that a $\Q^{k+1}_r$-version of $\widehat{f}^{(k+1,r)}_s$ is the \textit{static deflator} for $\mbox{cl}_{\Q^{k+1}_r}(\cC_s)$ with respect to $\Q^{k+1}_r$,  and that a $\Q^{k}_r$-version of $\widehat{f}^{(k,r)}_s$ satisfies that  $\widehat{f}^{(k,r)}_s \in \mbox{cl}_{\Q^{k}_r}(\cC_s) \subseteq \mbox{cl}_{\Q^{k+1}_r}(\cC_s)$, using that $\Q^{k+1}_r \ll \Q^k_r$, we obtain that 
	\begin{equation*}
	\E_{\Q^{k+1}_r}\Big[\tfrac{\widehat{f}^{(k,r)}_s}{\widehat{f}^{(k+1,r)}_s}\Big] \le 1.
	\end{equation*}
	 Hence, from the above estimates we can conclude that
	\begin{equation*}
	\E_{\P}\Big[ \tfrac{X_t/\widehat{f}^{(k+1,r)}_t}{X_s/\widehat{f}^{(k,r)}_s}\,\mathbbm{1}_{A^{k+1}_r} \Big] \le \P\big[A^{k+1}_r\big].
	\end{equation*}
Moreover, by  replacing $r$ with  $r + \frac{T}{2^{k+1}}$, we obtain with the same arguments that
\begin{equation*}
\E_{\P}\bigg[ \tfrac{X_t/\widehat{f}^{(k+1,r + \frac{T}{2^{k+1}})}_t}{X_s/\widehat{f}^{(k,r)}_s}\,\mathbbm{1}_{A^{k+1}_{r + \frac{T}{2^{k+1}}}} \bigg] \le \P\Big[A^{k+1}_{r + \frac{T}{2^{k+1}}}\Big].
\end{equation*}
Combining the above two bounds together with \eqref{eq: quotient between k+1 and k} implies that
	\begin{align*}
	\E_{\P}\Big[\tfrac{X_tZ^{k+1}_t}{X_sZ^k_s}\,\mathbbm{1}_{\{\tau > t\}}\Big] 
	&\le 
	\sum_{r \in \cD_k, t\le r<T} \Big(\P\big[A^{k+1}_r\big] + \P\big[A^{k+1}_{r + \frac{T}{2^{k+1}}}\big]\Big) \\
	&= \sum_{r \in \cD_k, t\le r<T} \P\big[A^k_r\big]\\
	&\le 1,
	\end{align*}
	as we claimed in \eqref{eq: estiamtes for P-expectation-red}. Finally, for $l \ge 1$,  we can  adapt the above proof using the disjoint decomposition of $A^k_r$ into $A^{k+l}_r \cup A^{k+l}_{r + \frac{T}{2^{k+l}}} \ldots \cup A^{k+l}_{r + \frac{(2^l - 1)T}{2^{k+l}}}$. Hence the claim is proved.

	Now for given $s<t$ in $\cD$ we fix $k_0 \in B_t \cap B_s$. Then for every $k \ge k_0$, the bound \eqref{eq: estiamtes for P-expectation} implies that 
	\begin{equation*}
	\E_{\P}\Big[\tfrac{X_t\,\text{\textbf{fconv}}_{p(s,t)}Z^k_t}{X_sZ^{k_0}_s}\Big] \le 1.
	\end{equation*}
	Letting $k \rightarrow \infty$, we obtain by Fatou's lemma that 
	$
	\E_{\P}\Big[\frac{X_tZ^\infty_t}{X_sZ^{k_0}_s}\Big] \le 1.
	$
	Then, by using the convexity of the function $(0,\infty)\ni x \mapsto 1/x$,  we also have that 
	\begin{equation*}
	\E_{\P}\Big[\tfrac{X_tZ^\infty_t}{X_s\text{\textbf{fconv}}_{p(s,t)}Z^{k_0}_s}\Big] \le 1.
	\end{equation*}
	Finally, by letting $k_0 \rightarrow \infty$ we conclude that indeed for all $s<t$ in $\cD$, for all $X \in \mathcal{X}$,
	\begin{equation}\label{eq: P-expectation of discounted processes2}
	\E_{\P}\Big[\tfrac{X_tZ^\infty_t}{X_sZ^\infty_s}\Big] \le 1.
	\end{equation}

	\textit{Step~3.2}: 
	(GSP on $\cD$ for conditional expectations).
	In this step we claim that for any $s<t \in \cD$ and any $X\in \cX$ the bound \eqref{eq: P-expectation of discounted processes2} also holds true for any conditional probability $\Q^{k_0}_q$, $k_0 \in B_t \cap B_s$, $q \in \cD_{k_0}\setminus\{T\}$, namely that
	\begin{equation}\label{eq: Q-expectation of discounted processes}
	\E_{\Q^{k_0}_q}\Big[\tfrac{X_tZ^\infty_t}{X_sZ^\infty_s}\Big] \le 1.
	\end{equation} 
	Indeed, to see this, observe first that if $q \le t$, then $X_t = 0$ under the measure $\Q^{k_0}_q := \P\big[\cdot\,|\, \tau \in (q, q+\frac{T}{2^{k_0}}]\big]$ and therefore  using Remark~\ref{remark: well defined condition} we see that \eqref{eq: Q-expectation of discounted processes} indeed holds in that case. It hence remains to consider the case where $q > t$.
	Following the same line as in \textit{Step~3.1}, we first consider the quantity $\E_{\Q^{k_0}_q} \Big[\frac{X_tZ^{k+1}_t}{X_sZ^k_s} \Big]$ for some $k \ge k_0$ and let
	\begin{equation*}
	 B_{q,k_0,k} := \Big\{r \in \cD_k: r \in (q, q+\tfrac{T}{2^{k_0}} - \tfrac{T}{2^k}]\Big\}.
	 \end{equation*}
	 In view of \eqref{eq: quotient between k+1 and k}, we can use \eqref{eq:3.1-3.2}  derived in the last step together with  the   disjoint decomposition $A^{k_0}_q = \bigcup_{r \in B_{q,k_0,k}} \big[A^{k+1}_r\cup A^{k+1}_{r+\frac{T}{2^{k+1}}}\big]$ to  check that
	\begin{align*}
	&\E_{\Q^{k_0}_q} \Big[\tfrac{X_tZ^{k+1}_t}{X_sZ^k_s} \Big] \\
	&= 
	\sum_{r \in B_{q,k_0,k}}\Bigg( \E_{\Q^{k+1}_r}\Big[\tfrac{X_t/\widehat{f}^{(k+1,r)}_t}{X_s/\widehat{f}^{(k,r)}_s}\Big]\,\tfrac{\P[A^{k+1}_r]}{\P[A^{k_0}_q]} 
	+
	\E_{\Q^{k+1}_{r + \frac{T}{2^{k+1}}}}\Big[\tfrac{X_t/\widehat{f}^{(k+1,r+\frac{1}{2^{k+1}})}_t}{X_s/\widehat{f}^{(k,r)}_s}\Big]
	\,\tfrac{\P[A^{k+1}_{r+\frac{T}{2^{k+1}}}]}{\P[A^{k_0}_q]}\Bigg) \\
	&\le 
	\sum_{r \in B_{q,k_0,k}} \Bigg(\tfrac{\P[A^{k+1}_r]}{\P[A^{k_0}_q]} + \tfrac{\P[A^{k+1}_{r+\frac{T}{2^{k+1}}}]}{\P[A^{k_0}_q]}\Bigg)\\
	& = 1.
	\end{align*}
	Finally, we can apply the same arguments used at the end  of \textit{Step~3.1} to see first for each fixed $k \ge k_0$ that $\E_{\Q^{k_0}_q} \Big[\tfrac{X_tZ^\infty_t}{X_sZ^k_s} \Big]\le 1$, and then by considering forward convex combinations of denominators $Z^k_s$ 
	that indeed the desired inequality \eqref{eq: Q-expectation of discounted processes} holds.\\

	\textit{ Step~4:} (Supermartingale property).
	
	\textit{ Step~4.1:} (Supermartingale property on $\cD$).
	Recall that by construction, see \eqref{eq: supermartingale deflator property globally and locally on Dk}, we have for every $k \in \N$ that the process $(Z^k_t)_{t\in \cD_k}$ is a strictly positive
	generalized supermartingale deflator for $\mathcal{X}^k$ with respect to $\P$. Consequently, by the same arguments as in the proof of Theorem~\ref{thm:discrete}, 
	we see that the process
	\begin{equation}\label{eq: adapted version of local deflators}
	\widetilde{Z}^k_t := \E_{\P}\big[Z^k_t \,\big|\, \cF_t\big], \quad t \in \cD_k,
	\end{equation}
	is a strictly positive 
	supermartingale deflator for $\mathcal{X}^k$ under the measure $\P$. In particular, it holds  for all $s<t \in \cD_k$ and $X \in \mathcal{X}$ that
	\begin{equation*}
	\E_{\P}\big[X_t\widetilde{Z}^k_t\,\big|\, \cF_s\big] \le X_s\widetilde{Z}^k_s.
	\end{equation*}
	Now fix $s<t$ in $\cD$. Then, observe that the above supermartingale property implies 
	for any $k \in B_s \cap B_t$ that
	\begin{equation}\label{eq:Step-4-1}
	\E_{\P}\Big[X_t\,\text{\textbf{fconv}}_{p(s,t)}\widetilde{Z}^k_t\,\Big|\, \cF_s\Big] \le X_s\,\text{\textbf{fconv}}_{p(s,t)}\widetilde{Z}^k_s.
	\end{equation}
	In addition, note that since the two sequences 
	\begin{equation*}
	\big(\text{\textbf{fconv}}_{p(s,t)}\widetilde{Z}^k_t\big)_{k \in B_s \cap B_t}
	 \quad \mbox{ and } \quad 
	\big(\text{\textbf{fconv}}_{p(s,t)}\widetilde{Z}^k_s\big)_{k \in B_s \cap B_t}
	\end{equation*}
	consists 
	of nonnegative random variables, we can find by Komlos lemma subsequences of forward convex combinations consisting of members from $\big(\text{\textbf{fconv}}_{p(s,t)}\widetilde{Z}^k_t\big)_{k \in B_s \cap B_t}$ and $\big(\text{\textbf{fconv}}_{p(s,t)}\widetilde{Z}^k_s\big)_{k \in B_s \cap B_t}$ 
	which we denote by 
	\begin{equation*}
	\big(\text{\textbf{fconv}}^{(2)}_{p(s,t)}\widetilde{Z}^k_t\big)_{k \in B_s \cap B_t}
	 \quad \mbox{ and } \quad 
	  \big(\text{\textbf{fconv}}^{(2)}_{p(s,t)}\widetilde{Z}^k_s\big)_{k \in B_s \cap B_t},
	\end{equation*}
	such that they converge to some nonnegative random variables $\widetilde{Z}^\infty_t$ and $\widetilde{Z}^\infty_s$ $\P$--a.s., respectively. Moreover, note that by Fatou's lemma and \eqref{eq:Step-4-1}, the supermartingale property is preserved by this limiting procedure, as 
	\begin{align*}
	\E_{\P}\Big[X_t\widetilde{Z}^\infty_t\Big| \cF_s\Big] 
	&\le
	 \liminf_{k \rightarrow \infty}\E_{\P}\Big[X_t\,\text{\textbf{fconv}}^{(2)}_{p(s,t)}\widetilde{Z}^k_t\,\Big|\, \cF_s\Big] \\
	&\le
	 \liminf_{k \rightarrow \infty}X_s\,\text{\textbf{fconv}}^{(2)}_{p(s,t)}\widetilde{Z}^k_s\\
	&= 
	X_s\widetilde{Z}^\infty_s.
	\end{align*}
	Therefore, as $s<t\in \cD$ was arbitrary, it remains to show that $\widetilde{Z}^\infty_s$ and $\widetilde{Z}^\infty_t$ are finite and strictly positive to conclude that $\widetilde{Z}^\infty$ is a strictly positive supermartingale deflator for $\cX^D$ on $\cD$. We focus on time $t$ as for $s<t$ the argument is the same.
	
	To see this, observe that since $\widetilde{Z}^k_t\, := \E_{\P}[Z^k_t \,|\, \cF_t]$, 
	we  have that
	\begin{equation*}
	\text{\textbf{fconv}}^{(2)}_{p(s,t)}\widetilde{Z}^k_t
	 =
	  \E_{\P}\Big[\text{\textbf{fconv}}^{(2)}_{p(s,t)}Z^k_t\,\Big|\, \cF_t\Big].
	\end{equation*}
	Moreover, recall that in \eqref{eq: cantor diagonal type convergence} of \textit{Step~2.2} we have found  $Z^\infty_t$ 
	which 
	is
	finite and strictly positive 
	such that  	$\P$--a.s.,
	\begin{equation*}
	Z^\infty_t = \lim_{k \rightarrow \infty} \text{\textbf{fconv}}_{p(s,t)}Z^k_t.
	\end{equation*}
	This, as the $\P$--almost surely type convergence is preserved by forward convex combinations, ensures that 	$\P$--a.s., also
	\begin{equation*}
	Z^\infty_t = \lim_{k \rightarrow \infty} \text{\textbf{fconv}}^{(2)}_{p(s,t)}Z^k_t.
\end{equation*}
	 As a consequence,  we have by Fatou's lemma that
	\begin{align*}
	\E_{\P}\Big[Z^\infty_t\Big| \cF_t\Big]
	 &= 
	 \E_{\P}\Big[ \lim_{k \rightarrow \infty}\, \text{\textbf{fconv}}^{(2)}_{p(s,t)}Z^k_t \,\Big|\,  \cF_t\Big] \\
	& \le
	 \liminf_{k \rightarrow \infty} \E_{\P}\Big[\text{\textbf{fconv}}^{(2)}_{p(s,t)} Z^k_t \, \Big| \,  \cF_t\Big] \\
	&=
	\liminf_{k \rightarrow \infty} \text{\textbf{fconv}}^{(2)}_{p(s,t)}\widetilde{Z}^k_t\\
	&= \widetilde{Z}^\infty_t.
	\end{align*}
	Since from \textit{Step~2.2} we know that $Z^\infty_t$ 
	is strictly positive and so 
	is its conditional expectation $\E_{\P}[Z^\infty_t\,|\, \cF_t]$, 
	we can conclude  thanks to the above inequality that indeed  $\widetilde{Z}^\infty_t$ 
	is also strictly positive. Finally, to see that $\widetilde{Z}_t$ is finite, note that on $\{\overline{X}_t = 0\}\subseteq \{\tau\leq t\}$, we have $Z^k_t=1$ for each $k$, see \eqref{eq: P deflators for kth dyadic points in continuous time}, which ensures that both $Z^\infty_t=1$ and $\widetilde{Z}^\infty_t=1$ on $\{\overline{X}_t = 0\}$.  Moreover, since $\E_{\P}[\overline{X}_t\widetilde{Z}^\infty_t] \le 1$ holds, 
	 we have $\widetilde{Z}^\infty_t < \infty$ also on $\{\overline{X}_t > 0\}$. This in turn shows that indeed $\widetilde{Z}_t$ is finite and we can conclude that indeed $(\widetilde{Z}^\infty_t)_{t\in \cD}$ is a strictly positive
	  supermartingale deflator for $\mathcal{X}^{\cD}:= \{(X_t)_{t \in \cD}: X \in \mathcal{X}\}$ with respect to $\P$. 
	  
	  \textit{ Step~4.2:} (Supermartingale property on [0,T]).
	  The extend $(\widetilde{Z}^\infty_t)_{t\in \cD}$ from $\cD$ to $[0,T]$ note that from the above step, we know that $S:=\overline{X}\widetilde{Z}^\infty$ is nonnegative supermartingale on $\cD$, hence by classical results, see for example \cite[Theorem~VI.2, p.67]{DellacherieMeyer.82}, we can extend $S$ from $\cD$ to $[0,T]$. Next, we can apply the same argument as in the proof of Theorem~\ref{theorem: second main theorem}, see \eqref{eq: regularized version of deflator}, to indeed obtain a strictly positive adapted c\`adl\`ag process $(Z_t)_{t \in [0,T]}$ such that $XZ$ is a supermartingale under $\P$ for all $X \in \mathcal{X}$.

	\textit{Step~5:} (A uniform bound for $1/Z^\infty$).
	In the last step we show the  uniform bound  \eqref{eq: uniform boundedness of deflator}  for $\frac{1}{Z^\infty_t}$ (for the dyadic case, the general case goes analogously), namely that for each $t \in \cD$ we have that 
	\begin{equation}\label{eq: uniform boundedness of deflator-dyadic}
	\lim_{M \to \infty} \sup_{k_0 \in B_t}\sup_{r \in \cD_{k_0}} \Q_r^{k}\big[\tfrac{1}{Z^\infty_t} \ge M\big] = 0.
	\end{equation}
	
	To that end, fix $t \in \cD$ and let $k_0 \in B_t$, $r \in \cD_{k_0}$.
	Now, let $(\text{\textbf{fconv}}_{t}Z^k_t)_{k \in B_t}$ be any sequence of forward convex combinations which converges to $Z^\infty_t$ $\P$--a.s., see also \textit{Step~2.2} for its existence, and denote 
	\begin{equation*}
	\text{\textbf{fconv}}_{t}Z^k_t := \lambda_1 Z^k_t + \ldots \lambda_{l+1}Z^{k+l}_t
	\end{equation*}
	where  $\lambda_1, \ldots, \lambda_{l+1}$ are the corresponding convex weights. 
	Invoking the concrete formula \eqref{eq: P deflators for kth dyadic points in continuous time} of each $Z^k_t$ and using a similar argument as in \textit{Step~2.1}, see \eqref{eq:Step2.1-1} and \eqref{eq:def-h-k}, but with respect to the following partition  $A^{k_0}_r=\cup_{i=0}^{(2^{k-k_0}-1)}\cup_{j=0}^{(2^{l}-1)} A^{k+l}_{r+\frac{iT}{2^{k}}+\frac{jT}{2^{k+l}}}$ of $A^{k_0}_r$, we see that
	\begin{equation*}
	\mathbbm{1}_{A^{k_0}_r}\,\frac{1}{\text{\textbf{fconv}}_{t} Z^k_t} 
	= 
	\sum_{i=0}^{2^{k-k_0}-1} \,\sum_{j=0}^{2^l-1} h^{(k+l,r+\frac{iT}{2^{k}}+\frac{jT}{2^{k+l}})}_t\,\mathbbm{1}_{A^{k+l}_{r+\frac{iT}{2^{k}}+\frac{jT}{2^{k+l}}}},
\end{equation*}
	where every $h^{(k+l,r+\frac{iT}{2^{k}}+\frac{jT}{2^{k+l}})}_t$ belongs to $\mbox{cl}_{\Q^{k+l}_{r +\frac{iT}{2^{k}}+ \frac{jT}{2^{k+l}}}}(\cC_t)$.
	Therefore, we deduce from the above identity that  for any $M>0$
	\begin{align*}
	&\Q^{k_0}_r\Big[\tfrac{1}{\text{\textbf{fconv}}_{t} Z^k_t} \ge M\Big] \\
	&=
	 \P\Big[\tfrac{1}{\text{\textbf{fconv}}_{t} Z^k_t} \ge M; A^{k_0}_r\Big]\frac{1}{\P[A^{k_0}_r]} \\
	&= 
	\sum_{i=0}^{2^{k-k_0}-1} \,\sum_{j=0}^{2^l-1} 
	\Q^{k+l}_{r +\frac{iT}{2^{k}}+ \frac{jT}{2^{k+l}}}\Big[h^{(k+l,r+\frac{iT}{2^{k}}+\frac{jT}{2^{k+l}})}_t\ge M\Big]\,\tfrac{\P\big[A^{k+l}_{r+\frac{iT}{2^{k}} + \frac{jT}{2^{k+l}}}\big]}{\P[A^{k_0}_r]}.
	\end{align*}
	Now, the uniform boundedness property \eqref{eq: uniform boundedness in probability} ensures for any  given $\varepsilon > 0$ that there exists an $M > 0$ such that for all $j, k,l$ and $r$ we have that
	\begin{equation*}
	 \Q^{k+l}_{r+\frac{iT}{2^{k}} + \frac{jT}{2^{k+l}}}\Big[h^{(k+l,r+\frac{iT}{2^{k}}+\frac{jT}{2^{k+l}})}_t\ge M\Big] \le \varepsilon.
	 \end{equation*} 
	 Combining this with the above equation ensures that for any such $M>0$
	\begin{equation*}
\Q^{k_0}_r\Big[\tfrac{1}{\text{\textbf{fconv}}_{t} Z^k_t} \ge M\Big]
 \le 
 \varepsilon 
 \sum_{i=0}^{2^{k-k_0}-1} \,\sum_{j=0}^{2^l-1} 
 \tfrac{\P\big[A^{k+l}_{r+\frac{iT}{2^{k}} + \frac{jT}{2^{k+l}}}\big]}{\P[A^{k_0}_r]} = \varepsilon.
	\end{equation*}
	By passing to the limit, we hence also get $\Q^{k_0}_r\Big[\frac{1}{Z^\infty_t} \ge M\Big] \le \varepsilon$, which in turn indeed implies \eqref{eq: uniform boundedness of deflator-dyadic}  and finishes the proof.
	\end{proof}
\begin{remark}\label{remark: convex combination of generalized supermartingale is not a generalized supermartingale}
	Compared to supermartingale deflators, one cannot expect that the convex combination of two generalized supermartingale deflators is again a generalized supermartingale deflator. Therefore, we could not directly construct $Z^\infty$ from convex combinations of 
	generalized supermartingale deflators on $\cD_k$ to obtain  a
	generalized supermartingale deflator on $\cD$. 
\end{remark}

\subsubsection{Proof of Theorem~\ref{theorem: independent clock}}\label{subsubsec:independent-clock-thm}
In this subsubsection we provide the proof of Theorem~\ref{theorem: independent clock}~and hence all the corresponding assumptions are in force.
\begin{proof}
	The implication $(ii) \Rightarrow (i)$ follows by the same argument as in the proofs presented before, hence it remains to show the implication $(i) \Rightarrow (ii)$.
	To that end, let $\cD$ consists of all the dyadic numbers on $[0,T]$, for each $k\in \N$ let $\cD_k:=\{\nicefrac{iT}{2^k}\colon i:=0,1,\dots,2^k\}$ be the collection of all $k$-th dyadic numbers in $[0,T]$,
		and 
		let $\mathcal{X}^k:=\{(X_t)_{t \in \cD_k}: X \in \mathcal{X}\}$ be the restriction of the market to the $k$-dyadic grid.
	In the spirit of the last subsubsection, we define 
	for any $r \in \cD_k\setminus\{T\}$ 
	here 
	\begin{equation*}
		\Q^k_r[\,\cdot\,] := \P\big[\cdot \,\big| \,\widetilde\tau \in (r, r+\tfrac{T}{2^k}]\big].
	\end{equation*}
	
	Now observe first that
	thanks to the property of $\widetilde\tau$, the market $\mathcal{X}$ contains a num\'eraire under $\Q^k_r$ until time $r$ and all elements in $\mathcal{X}$ vanish after time $r + \frac{T}{2^k}$. Hence, using exactly the same argument as in the  proof of \textit{Step~1} in Theorem~\ref{theorem: third main theorem}, we know that for each $k$ 
	the process $(Z^k_t)_{t \in \cD_k}$ defined by 
	\begin{equation*}
		Z^k_t = \sum_{r \in \cD_k, t\le r<T} \frac{1}{\widehat{f}^{(k,r)}_t}\,\mathbbm{1}_{\{\widetilde\tau \in (r, r+ \frac{T}{2^k}]\}} + \mathbbm{1}_{\{\widetilde\tau \le t\}}, \quad t \in \cD_k,
	\end{equation*}
	is a strictly positive generalized supermartingale deflator for $\mathcal{X}^k$ defined on $\cD_k$, where by using Lemma~\ref{lemma: nice version of optimizer} we pick $\widehat{f}^{(k,r)}_t$
	to be an $\cF_t$--measurable $\Q^k_r$-version of the \textit{static deflator}  of $\mbox{cl}_{\Q^k_r}(\cC_t)$.
	We  claim that for each $t\in \cD$ the set
	\begin{equation*}
		\mathcal K_t := \text{conv}\Big\{\tfrac{1}{\E_{\P}[Z^k_t |\cF_t]}: k \in B_t\Big\}
	\end{equation*}
	(recall that $B_t := \{k \in \N: t \in \cD_k\}$) is $\P$-bounded. 
	
	Indeed, to see this, note that for each $k\in B_t$ the identity
	\begin{equation*}
		\frac{1}{Z^k_t} =  \sum_{r \in \cD_k, t\le r<T} \widehat{f}^{(k,r)}_t\,\mathbbm{1}_{\{\widetilde\tau \in (r, r+ \frac{T}{2^k}]\}} + \mathbbm{1}_{\{\widetilde\tau \le t\}},
	\end{equation*}
	Jensen's inequality, and the $\cF_t$--measurability of every $\widehat{f}^{(k,r)}_t$ imply that 
	\begin{equation*}
		\frac{1}{\E_{\P}[Z^k_t \,|\,\cF_t]} \le \E_{\P}\big[\tfrac{1}{Z^k_t} \,\big|\,\cF_t\big] \le \sum_{r \in \cD_k, t\le r<T} \widehat{f}^{(k,r)}_t\,\E_{\P}\big[\mathbbm{1}_{\{\widetilde\tau \in (r, r+ \frac{T}{2^k}]\}}\, \big|\,\cF_t\big] + 1.
	\end{equation*}
	In addition, as by assumption 
	the event $\big\{\widetilde\tau \in (r, r+ \tfrac{T}{2^k}]\big\}$ is independent of $\cF_r$,
	we have for all $r \in \cD_k \setminus\{T\}$ with $r \geq t$ that
	\begin{equation*}
		\E_{\P}\Big[\mathbbm{1}_{\{\widetilde\tau \in (r, r+ \tfrac{T}{2^k}]\}}\,\Big|\,\cF_t\Big] 
		= 
		\P\big[\widetilde\tau \in (r, r+ \tfrac{T}{2^k}]\big],
	\end{equation*}
	which implies that
	\begin{equation}\label{eq:5-1}
		\frac{1}{\E_{\P}[Z^k_t |\cF_t]} \le \sum_{r \in \cD_k, t\le r<T} \widehat{f}^{(k,r)}_t \,\P\big[\widetilde\tau \in (r, r+ \tfrac{T}{2^k}]\big] + 1.
	\end{equation}
	Moreover, note that from Lemma~\ref{lemma: nice version of optimizer} 
	we know that $\widehat{f}^{(k,r)}_t$ can be taken from the $\P$--closure of the $\P$-solid hull of $\cC_t$. Furthermore, since $\cC_t$ is convex, also its $\P$-solid hull and the $\P$-closure of its $\P$-solid hull are convex.
	This together with the estimate  $\sum_{r \in \cD_k, t\le r<T} \P\big[\widetilde\tau \in (r, r+ \frac{T}{2^k}]\big] \le \P[\widetilde\tau > t] \le 1$ shows that for each $k\in B_t$ the term
	\begin{equation*}
		\sum_{r \in \cD_k, t\le r<T} \widehat{f}^{(k,r)}_t \,\P\big[\widetilde\tau \in (r, r+ \tfrac{T}{2^k}]\big]
	\end{equation*}
	belongs to the $\P$-closure of the $\P$-solid hull of $\cC_t$. 
	Moreover, recall that the $\P$-closure of the $\P$-solid hull of $\cC_t$ is $\P$-bounded thanks to the $\P$-boundedness $\cC_t$, see Lemma~\ref{lemma: closure and closure of solid hull is bounded}, as by assumption $(i)$ NUPBR$_t$ holds for each $t$.
	Therefore, we can conclude from \eqref{eq:5-1} that indeed, the set
	\begin{equation*}
		\mathcal K_t := \text{conv}\Big\{\tfrac{1}{\E_{\P}[Z^k_t |\cF_t]}: k \in B_t\Big\}
	\end{equation*}
	is $\P$-bounded.
	
	Next for each $t\in \cD$ we apply Komlos lemma for the sequence $\E_{\P}[Z^k_t \,|\,\cF_t]$, $k \in B_t$, to obtain a nonnegative limit $Z_t$. Note that by the same argument as at the beginning of \textit{Step~2.2} in the proof of Theorem~\ref{theorem: third main theorem}, we see that the sequence $\E_{\P}[Z^k_t \,|\,\cF_t]$, $k \in B_t$, is $\P$-bounded, and hence $Z_t<\infty\ \P$-a.s.. Moreover, the above derived $\P$-boundedness of $\cK_t$ ensures that $\frac{1}{Z_t}$ is also finite $\P$--a.s., which means that $Z_t$ is strictly positive. Since for every $k$, we have that $\E_{\P}[Z^k_t\,|\, \cF_t]$, $t \in \cD_k$, is a 
	strictly positive
	supermartingale deflator, so is any forward convex combination.
	This in turn ensures that $(Z_t)_{t\in \cD}$ is 
	a strictly positive supermartingale deflator for $\mathcal{X}^{\cD}$ under $\P$ on all dyadic numbers. Then, the same argument as in \textit{Step~4.2} in the proof of Theorem~\ref{theorem: third main theorem} allows us to extend $(Z_t)$ from $\cD$ to  $[0,T]$ such that it becomes a
	strictly positive, c\`adl\`ag supermartingale deflator on $[0,T]$.
\end{proof}
\section{Appendix}
Let us provide the proof of the statement in Remark~\ref{remark: compare with the classical fork convexity} and Lemma~\ref{le:remark: compare with the classical fork convexity} that
in the presence of a num\'eraire, the property for a market to satisfy NUPBR does not depend on the choice of the definition of the fork-convexity (between the one of {\v{Z}}itkovi{\'c} \cite{Zitkovic.02} and ours).
\begin{proof}[Proof of Lemma~\ref{le:remark: compare with the classical fork convexity}]
 First, if the fork convex hull of the market $\mathcal{X}$ satisfies NUPBR, then so does $\mathcal{X}$ being a subset of its hull. On the other hand, suppose that $\mathcal{X}$ satisfies the NUPBR condition. We need to show that the fork-convex hull taken with respect to our notion also satisfies NUPBR. Note that by classical arguments (see, e.g., the beginning of the proof of Theorem~\ref{thm:discrete}) it suffices to prove the existence a strictly positive {(generalized) supermartingale deflator} for the fork-convex hull to guarantee that it satisfies NUPBR. 
 
 To that end, observe that since by assumption $\cX$ is $\mathbb{F}$--adapted and is fork convex in the sense of {\v{Z}}itkovi{\'c} \cite{Zitkovic.02} which is also used in Kardaras \cite{Kardaras.13.2},  we can apply his result \cite[Theorem~2.3]{Kardaras.13.2}, to guarantee 
 the existence of a strictly positive supermartingale deflator $Z$ for $\mathcal{X}$. We claim that $Z$ is also a strictly positive  supermatingale deflator with respect to the fork-convex hull of $\cX$. To see this, let   $X^1, X^2, X^3 \in \mathcal{X}$ and let $X$ be defined as in \eqref{eq: fork convexity in no numeraire case}. Then we have for $t \ge s$ and $A \in \cF_s$ that
\begin{align*}
\E_{\P}\Big[\tfrac{X_tZ_t}{X_sZ_s}\Big| \cF_s\Big]
 &= \E_{\P}\bigg[\tfrac{\mathbbm{1}_A\big[\big(\frac{X^2_t}{X^2_s}X^1_s\big)\,\mathbbm{1}_{\{X^2_s>0\}} + X^1_t\,\mathbbm{1}_{\{X^2_s=0\}}\big]Z_t}{X^1_sZ_s}\,\bigg|\, \cF_s\bigg] \\
& \quad
+ \E_{\P}\bigg[\frac{\mathbbm{1}_{A^c}\big[\big(\frac{X^3_t}{X^3_s}X^1_s\big)\,\mathbbm{1}_{\{X^3_s>0\}} + X^1_t\,\mathbbm{1}_{\{X^3_s=0\}}\big]Z_t}{X^1_sZ_s}\,\bigg|\, \cF_s\bigg].
\end{align*}
For the first term on the right-hand-side of the above equation, we can  use the facts that $X^iZ$ are generalized supermartingales for $i=1,2$, and  that $\{X^2_s \ge 0\}$, $\{X^2_s = 0\}$ are $\cF_s$--measurable to obtain that
\begin{align*}
&\E_{\P}\bigg[\tfrac{\mathbbm{1}_A\big[\big(\frac{X^2_t}{X^2_s}X^1_s\big)\,\mathbbm{1}_{\{X^2_s>0\}} + X^1_t\,\mathbbm{1}_{\{X^2_s=0\}}\big]Z_t}{X^1_sZ_s}\,\bigg|\, \cF_s\bigg]\\
 & = \E_{\P}\Big[\tfrac{X^2_tZ_t}{X^2_sZ_s}\Big| \cF_s\Big]\,\mathbbm{1}_A \, \mathbbm{1}_{\{X^2_s>0\}} 
+ \E_{\P}\Big[\tfrac{X^1_tZ_t}{X^1_sZ_s}\Big| \cF_s\Big]\,\mathbbm{1}_A\,\mathbbm{1}_{\{X^2_s=0\}} \\
&\le \mathbbm{1}_A \, \mathbbm{1}_{\{X^2_s>0\}}
 + \mathbbm{1}_A \, \mathbbm{1}_{\{X^2_s=0\}}\\
& = \mathbbm{1}_A.
\end{align*}
Similarly, we can show that the second term on the right-hand-side of the above equation is bounded by $\mathbbm{1}_{A^c}$, and hence we get
\begin{equation*}
\E_{\P}\Big[\tfrac{X_tZ_t}{X_sZ_s}\Big| \cF_s\Big] \le 1.
\end{equation*}
This in turn, ensures that $Z$ is indeed a strictly positive supermartingale deflator for the fork convex hull of $\cX$, which  is sufficient to guarantee that the fork convex hull of $\cX$ satisfies NUPBR.  {For more details we refer readers to the proof of \cite[Proposition~3]{Zitkovic.02}}.
\end{proof}
\vspace{2em}
\noindent
{\bf Acknowledgment:}
The authors would like to thank D\'aniel B\'alint and Josef Teichmann for fruitful discussions.\\
{The second author gratefully acknowledges the financial support by the SNSF Grant P2EZP2$\_$188068.}\\
{The third author gratefully acknowledges the financial support by his Nanyang Assistant Professorship Grant (NAP Grant) \textit{Machine Learning based Algorithms in Finance and Insurance}.}

\newpage
\newcommand{\dummy}[1]{}

\bibliographystyle{plain}

\end{document}